\documentclass[aps,showpacs,superscriptaddress,floatfix,nofootinbib,11pt]{article}

\usepackage[margin=2.8cm]{geometry}
\usepackage{amsfonts}
\usepackage{graphicx}
\usepackage{caption}
\usepackage{thmtools,thm-restate}
\usepackage{amssymb}
\usepackage{amsmath}
\usepackage{amsthm}
\usepackage{mathrsfs}
\usepackage{bm}
\usepackage{gensymb}
\usepackage{algorithm}
\usepackage{algpseudocode}
\usepackage{booktabs}
\usepackage{setspace}
\usepackage[braket,qm]{qcircuit}
\usepackage{mathtools}
\usepackage{stmaryrd} 
\usepackage[pdftex,colorlinks=true,linkcolor=black,citecolor=blue,urlcolor=black]{hyperref}	
\usepackage[table,xcdraw]{xcolor}
\usepackage{adjustbox}
\usepackage{makecell}
\usepackage{caption}
\usepackage{subcaption}
\usepackage{float}
\newfloat{algorithm}{t}{lop}
\graphicspath{ {./figures/} }
\usepackage{multirow}

\usepackage{placeins}

\usepackage{mathtools}

\makeatletter
\AtBeginDocument{%
  \expandafter\renewcommand\expandafter\subsection\expandafter
    {\expandafter\@fb@secFB\subsection}%
  \newcommand\@fb@secFB{\FloatBarrier
    \gdef\@fb@afterHHook{\@fb@topbarrier \gdef\@fb@afterHHook{}}}%
  \g@addto@macro\@afterheading{\@fb@afterHHook}%
  \gdef\@fb@afterHHook{}%
}
\makeatother

\newcommand{\later}[1]{}

\newcommand{\cA}{\mathcal{A}}
\newcommand{\cQ}{\mathcal{Q}}

\newcommand{\mA}{\mathcal{A}} 
 
\newcommand{\mH}{\mathcal{H}}
\newcommand{\mO}{\mathcal{O}}

\newcommand{\fn}{\mathfrak{n}}

\newcommand{\poly}{\mathop{poly}}

\DeclareMathOperator*{\argmax}{arg\,max}
\DeclareMathOperator{\avg}{avg}

\newtheorem{theorem}{Theorem}
\newtheorem{lemma}{Lemma}

\DeclarePairedDelimiter\ceil{\lceil}{\rceil}

\newcommand{\ccb}{\cellcolor[HTML]{DAE8FC}}
\newcommand{\ccr}{\cellcolor[HTML]{FFCCC9}}
\newcommand{\ccrr}{\cellcolor[HTML]{FF8F8C}}
\newcommand{\ccrrr}{\cellcolor[HTML]{F35955}}
\newcommand{\ccg}{\cellcolor[HTML]{9AFF99}}
\newcommand{\ccgg}{\cellcolor[HTML]{87DC86}}
\newcommand{\ccggg}{\cellcolor[HTML]{4EBC4E}}

\newcommand{\ZQ}{$\textbf{QSearch}_{\text{Zalka}}$ }

\newcommand{\qmi}{\textbf{QMax}_{\infty}}

\begin{document}

\setstretch{1}

\title{Quantum Algorithms for Community Detection and their Empirical Run-times}
\author{Chris Cade, Marten Folkertsma, Ido Niesen, and Jordi Weggemans}
\date{\today}

\captionsetup{width=400pt}

\date{\today}

\maketitle

\begin{abstract}
\noindent We apply our recent work~\cite{ourotherpaper} on empirical estimates of quantum speedups to the practical task of community detection in complex networks. We design several quantum variants of a popular classical algorithm -- the \textit{Louvain algorithm} for community detection -- and first study their complexities in the usual way, before analysing their complexities empirically across a variety of artificial and real inputs. We find that this analysis yields insights not available to us via the asymptotic analysis, further emphasising the utility in such an empirical approach. In particular, we observe that a complicated quantum algorithm with a large asymptotic speedup might not be the fastest algorithm in practice, and that a simple quantum algorithm with a modest speedup might in fact be the one that performs best. Moreover, we repeatedly find that overheads such as those arising from the need to amplify the success probabilities of quantum sub-routines such as Grover search can nullify any speedup that might have been suggested by a theoretical worst- or expected-case analysis. 
\end{abstract}

\setcounter{tocdepth}{2}
\tableofcontents

\section{Introduction}\label{sec:introduction}
Estimating the impact that a quantum computer could have for a given computational problem requires an honest assessment of the potential improvement in speed\footnote{or, indeed, accuracy.} that a quantum algorithm might achieve over the state-of-the-art classical one. This is a tricky task made harder by the current level of maturity of quantum hardware: quantum algorithms can only be implemented for very small problem instances, and even then the output is so marred by noise that it is often difficult to assess even the correctness of the computation. As such, one often resorts to theoretical analyses of quantum algorithms and proves, rigorously, that they are likely to achieve a speedup over an equivalent classical algorithm, assuming perhaps that overheads such as those from error correction are suitably modest. Such analyses usually yield upper bounds on the worst-case run-times of the quantum algorithms -- very often for artificially constructed, `difficult' problem instances -- and a speedup is concluded whenever these bounds scale better than those obtained via an analysis of the best-known classical algorithm. 

Whilst this approach offers valuable insight, it does not tell the entire story. For instance, it can be the case that a provable run-time is not available, or that it is but ends up being uninformative, as is commonly the case for heuristic algorithms that have large worst-case run-times, but often perform well on instance of practical interest. Moreover, for a particular computational problem it might be possible to design several variants of a quantum algorithm, all with run-times that scale differently in different situations, and it can be difficult to decide which one will be fastest in practice. Often the performance of such algorithms will depend on the particular inputs on which they are run, something that can be somewhat difficult to account for in a mathematical study. Moreover, it can be the case that a quantum version of a classical algorithm only speeds-up part of the algorithm, and how much of an overall speedup can be obtained is input dependent, and not clear from the asymptotic behaviour. 

As such, it is likely that it will be necessary to study the performance of quantum algorithms from a more empirical standpoint in order to assess their usefulness in the time before large, fault tolerant quantum computers become widespread. In~\cite{ourotherpaper}, we suggested a general framework for doing so that was based on the combination of carefully derived upper bounds on the complexities of quantum sub-routines with classical simulation of the entire algorithm, in a way that allowed for the estimation of the quantum run-time. We gave evidence for the utility of this approach by studying the potential quantum speed-ups that might arise from quantum versions of classical heuristic algorithms for solving {\sc maxsat}, an optimisation problem that generalises the {\sc sat} problem. 

In this paper, we apply the methodology developed in~\cite{ourotherpaper} to a problem of more practical interest. We consider quantum speedups of a popular heuristic algorithm that goes by the name of the \emph{Louvain algorithm}\footnote{The algorithm takes the name of the city in which it was developed. The original paper describing the method has been cited over 15,000 times, and the algorithm itself can be found in all popular graph/network analysis software packages.} which forms one of the main tools for tackling a problem ubiquitous in the study of complex networks: that of \emph{community detection}. Together with its descendants, the Louvain algorithm has successfully been used to study large sparse networks with millions of vertices~\cite{blondel08,de2011generalized,que2015scalable,Traag2019}. Taking this as a use-case, we demonstrate further the usefulness of such a `semi-empirical' approach in both estimating the potential for quantum speedups, as well as in the \textit{design} of quantum algorithms for a particular problem. In particular, we show how a numerical study can unveil significant performance differences between various quantum algorithms for the same problem, that were not obvious a priori from an asymptotic analysis alone.

\paragraph{Community detection}
One of the main topics in the study of complex networks is whether the nodes in the network form densely connected clusters, called communities. Uncovering the community structure of a network allows for a better understanding of the network as a whole. The task of partitioning the network into communities is known as community detection. Community detection plays an important role in a variety of topics, such as (but by no means limited to) social networks~\cite{Ozer16}, recommendation systems~\cite{Ahmed13}, E-commerce~\cite{Rios14}, scientometrics~\cite{Ji16}, biological systems and healthcare~\cite{Wang2005}, and economics~\cite{Gui2014}.

Community detection falls under the broader category of graph partitioning: formally, community detection is the task of partitioning the vertex set of a graph by maximizing a particular function that expresses the quality of the partition. The most commonly used metric is the \emph{modularity}, which assigns a value between -1 and 1 to each partition of the vertex set of the graph~\cite{Girvan02} by assigning an effective weight\footnote{The actual edge weight minus the expected weight in the so-called \emph{configuration model}.} to each edge and then summing the effective weights of all edges connecting vertices that are in the same community. The number of communities is not fixed ahead of time, and hence finding how many partitions are needed to maximise the modularity is included as a part of the problem.

Complexity-wise, the problem of finding the partition that exactly maximises the modularity is an $\mathsf{NP}$-hard problem~\cite{Brandes08}, and therefore it is common (indeed, necessary) to resort to heuristic methods for community detection. Commonly used heuristics are those based on hierarchical agglomeration~\cite{Clauset04} extremal optimisation~\cite{Duch05}, simulated annealing ~\cite{Reichardt06,Guimera05}, spectral algorithms~\cite{Newman06}, and the Louvain method~\cite{blondel08}. The Louvain algorithm has been found to be one of the fastest and best performing algorithms in various comparative analyses \cite{Lancichinetti09,Yang08}, and it is this algorithm that we choose as the basis for our quantum algorithms for community detection.

\

\subsection{Summary of results}
We design several quantum algorithms for community detection on graphs by building quantum versions of the Louvain algorithm, and analyse their complexities in the usual way, finding, as is not uncommon for heuristic quantum algorithms, a per-step speedup. We then apply the bounds and methodology from~\cite{ourotherpaper} to estimate the complexities of the algorithms on practical inputs, and investigate whether the speedups promised by the asymptotic analyses manifest in practice. 

Our results are summarised in Table~\ref{tab:summary_results}, in which we compare the expected complexities of the classical Louvain algorithm, as well as three quantum variants of it: `\textbf{QLouvain}' (a direct speedup of the classical algorithm), `\textbf{SimpleQLouvain}' (a simplification of the preceding algorithm), and `\textbf{EdgeQLouvain}' (a quantum version of a significant simplification of the Louvain algorithm). We also consider `sparse graph' versions of the algorithms (indicated by a `\textbf{SG}' suffix)), whose asymptotic run-times are worse than their original versions, but whose run-times in practice are likely to be faster when the input graph is sparse. The two right-most columns provide information about the \textit{empirically observed} speedups obtained by the algorithms when they were simulated, in the sense discussed above. 

\begin{table}[htb]
\centering
\begin{tabular}{l|l|l|l}
 & \makecell{Query complexity \\ per step $k$} &   \makecell{Absolute\\ speed-up \\ observed?} &   \makecell{Empirically observed\\ range of polynomial \\ speed-ups}\\ \hline
 \makecell{\textbf{Louvain}} & $ \mO\left( \delta_{\max}\tau_k \right)$ & -- & -- \\ \hline
\makecell{\textbf{QLouvain}}& $\tilde{\mO}(\sqrt{\delta_{\max}\tau_k} )$ & \ccrr No  &  \ccrr $0.85$ - $0.99$  \\ 
\makecell{\textbf{QLouvainSG}}& $\tilde{\mO}( \delta_{\max}\sqrt{\tau_k} )$ &\ccrr No  &  \ccrr $0.70$ - $0.86$  \\ 
\makecell{\textbf{SimpleQLouvain}} & $\tilde{\mO}\left(\sqrt{\frac{\delta_{\max}}{f_k}} \right)$  & \ccrr No & \ccgg $1.04$ - $1.25$ \\ 
\makecell{\textbf{SimpleQLouvainSG}} & $\tilde{\mO}({\frac{\delta_{\max}}{\sqrt{f_k}}} )$  & \ccrr No & \ccgg $1.13$ - $1.55$ \\ 
\makecell{\textbf{VTAA QLouvain}}  & $\tilde{\mO} (t_\text{avg}^q \sqrt{\tau_k} )$  & --  & --\\ 
\makecell{\textbf{EdgeQLouvain}} & $\tilde{\mO} \left(\frac{1}{\sqrt{h_k}} \right)$  &  \ccgg Yes  & \ccgg $1.18$ - $1.49$ \\ 
\end{tabular}
\caption{Overview of the main results obtained by applying our techniques to the proposed quantum versions of the Louvain algorithm. The second column shows upper bounds on the expected number of queries when performing a single step. The third column indicates whether we observed an absolute query count speed-up by the quantum algorithm over Louvain on our artificially generated networks up to size $n=10^5$, and the fourth column shows the estimated range of polynomial speed-ups based on the same data. Here, $\delta_{\max}$ is the maximum number of communities adjacent to any single vertex, $f_k$ is the fraction of vertices in the graph that are `good' during step $k$, $\tau_k$ is the number of vertices inspected by the classical algorithm during step $k$, $h_k$ is the fraction of \textit{edges} and $\tilde{h}_k$ the fraction of node-neighbouring community pairs that yield good moves for vertices during step $k$. $t_\text{avg}^q$ is defined in Section~\ref{sec:QVTAA}. For a definition of all these terms, we defer to Section~\ref{sec:quantum} of the paper.}
\label{tab:summary_results}
\end{table}

One source of overhead that we find contributed significantly to the overall complexities were logarithmic overheads due to success probability amplification of subroutines -- for the algorithms to work correctly, we often require that \emph{all} calls to quantum subroutines succeed with high probability, which often yields quite large overheads in practice. Because of this, we found that in general the more complicated quantum algorithms offered less of a speedup (or none at all) in practice, despite indicating a generic square-root speedup per step over the original classical algorithm. These observations suggest that `greedily' favouring a larger asymptotic speedup might actually lead to \emph{slower} run-times in practice, and that a more nuanced analysis is required if we are to maximise quantum speedups in practice. 

For us, these findings underscore the need to investigate and consider the \textit{actual} quantum speedup that might be achieved in practice on realistic data sets, rather than concluding that a speedup will be obtained from an asymptotic analysis alone. As we show in this paper, such an approach can also be quite useful for comparing different quantum versions of the same algorithm, something that could facilitate the future design of quantum algorithms for practical tasks. 

\subsection{Methodology}
In~\cite{ourotherpaper}, we introduced the necessary tools and methodology for obtaining accurate numerical estimates of the complexities of quantum algorithms with a reasonably generic form that is common to many quantum speedups of classical heuristic algorithms. In particular, we considered algorithms with the form shown in Algorithm~\ref{alg:general}. 
\begin{algorithm}[!htb]
  \caption{Generic quantum algorithm structure}
  \label{alg:general}
\begin{algorithmic}[1]
    \State \textbf{Input} $X$, \textbf{Memory} $M$
    \For{$k = 1, \dots, T$}
        \State Do some classical processing on $X$ and $M$, resulting in some list $L_k$ containing $t_k$ marked items.
        \State Perform either one or more (perhaps nested) Grover searches with an unknown number of marked items on $L_k$, or run quantum maximum-finding on the list $L_k$, to obtain some item $x_k$.
        \State Do some more classical processing given $x_k$, update $M$.
    \EndFor
\end{algorithmic}
\end{algorithm}
Our approach was to run the algorithm classically, by replacing Step 4 (the call to a quantum sub-routine) with a classical procedure that gives the same output behaviour whilst collecting the information required to estimate what the quantum run-time complexity would have been if it had been used. The quantum complexity estimates themselves were obtained via tight bounds, including all constants, of two important sub-routines: Grover search with an unknown number of marked items, and quantum maximum finding. In cases where the information required to calculate the quantum complexities could not be computed exactly (e.g. because the input sizes were too large to make a such a computation infeasible), we gave methods for estimating them whilst retaining guarantees on the complexities produced. 

In all cases, `complexity' refers to a particular choice of measure, which for us was (and will be) the number of times a particular function is called by the classical or quantum algorithm. Of course, this does not represent the true run-time, and in particular does not include overheads such as those from quantum error correction. Nevertheless, this choice of complexity enables a clean comparison between classical and quantum algorithms, as well as between different quantum algorithms for the same task. It might be that a quantum speedup suggested by our (empirical) analysis does not manifest in practice due to such overheads, but that is not our main focus -- our goal is to study whether a quantum speedup could manifest \textit{at all}, even assuming zero overhead from the likes of error correction or noise. If the algorithms fall short at this level of analysis, then a quantum speedup can already be ruled out without taking the time and effort to compile the quantum algorithm for a particular piece of hardware. In addition, this complexity measure is independent of the details of quantum hardware, which is likely to change over the coming years.

\subsection*{Organization}
In Section~\ref{sec:prelims} we introduce the practical task of community detection, describe the popular classical Louvain algorithm for it (Section~\ref{sec:Louvain}), and analyse its asymptotic complexity (Section~\ref{sec:classical}). In Section~\ref{sec:quantum} we proceed to construct several quantum variants of the Louvain algorithm, with the aim of comparing the classical and quantum performances empirically. In Section~\ref{part:three}, we give tight bounds on the complexities of our main quantum sub-routines (Section~\ref{sec:complexity_bounds}), and describe our approach to simulating these quantum algorithms (Section~\ref{sec:simulation_details}). Finally, in Section~\ref{sec:numerics} we analyse their complexities numerically, comparing their performances to the original classical algorithm, and amongst each other.

Appendix~\ref{app:number_of_moves_main_section} discusses the number of vertices moved by the Louvain algorithm, both from a theoretical and an empirical perspective; Appendix~\ref{sec:runtime_OL_LL} provides some numerical results comparing the performance of the original Louvain algorithm with our implementation thereof that contains and additional data structure; Appendix~\ref{app:VTAA} gives details of a slightly more efficient quantum algorithm for community detection based on the technique of variable time amplitude amplification; and finally Appendix~\ref{sec:random_graphs} describes the algorithm we use to generate FCS-type random graphs.

\section{Community detection}
\label{sec:prelims}

In this section, we formally introduce the problem of community detection in graphs and describe the  Louvain algorithm. We begin by introducing some notation, and then proceed to define the \textit{modularity} function, which serves as a measure of quality for community assignments, before describing the Louvain algorithm itself.

\paragraph{Notation} In this manuscript, $G = (V,E)$ is a graph with vertex set $V$ and edge set $E$. We write $n = |V|$ for the number of vertices and denote the $n \times n$ (weighted) adjacency matrix of the graph by $A$, which we assume to be symmetric, $A = A^T$, real-valued with non-negative entries, and without self-loops: $A_{vv} = 0$ for all $v\in V$. We write $d_u$ for the degree of a vertex $u \in V$ and $s_u = \sum_{v} A_{uv}$ for the \emph{strength} of vertex $u$, defined as sum of the weights of all edges incident to $u$. We denote the neighborhood of $u$ by $N_u := \{v\in V: A_{uv} > 0\}$. Furthermore, write $d_{\max} = \max_{u\in V} d_u$ for the maximum degree, and let $W = \frac{1}{2}\sum_{uv} A_{uv}$ be the sum of all weights. Finally, for any positive integer $k$, we write $[k]$ for the set $\{1,\dots,k\}$.

\paragraph{Access to the input graph}
We assume we have \textit{adjacency list access} to the graph $G$. That is, for each $u\in V$, we have access to the list of neighbors of $u$ through the function $\fn_u : [d_u] \rightarrow V$. Specifically, given $j\in [d_u]$, we can query the $j$-th neighbor $\fn_u(j) \in N_u \subseteq V$ of $u$, as well as the weight $A_{u\fn_u(j)}$ on the edge connecting $u$ and $\fn_u(j)$. We assume that we know the degrees of each vertex ahead of time, or otherwise that we compute them during the pre-processing step (see below).

Finally, we assume that the vertices have some arbitrary but fixed ordering, and that the adjacency lists are sorted according to this ordering, so that, given any vertex $u\in V$ and a neighbor $v \in N_u$, we can in $\mO(\log d_u)$ time find the index $j\in [d_u]$ such that $\fn_u(j) = v$ using binary search. If the adjacency lists are not sorted when they are given to use, then instead we can sort them all in time $\tilde{O}(nd_{\max})$ before continuing.

\subsection{Modularity}
Formally, a community partitioning of $V$ is given by a label function $\ell:V \rightarrow [n]$ that assigns to every vertex $v\in V$ a label $\ell(v) \in [n]$. All vertices with the same label are said to be in the same community, and we denote the community of a given vertex $v\in V$ by $C_{\ell(v)} \subset V$. Likewise, for any label $\alpha \in [n]$, $C_\alpha := \ell^{-1}(\alpha)$ denotes the set of all vertices contained in the community labelled $\alpha$. For clarity we will use Roman characters (e.g. $u$, $v$) to refer to vertices, and Greek letters (e.g. $\alpha$, $\beta$) for community labels. 

\ 

\noindent Given a community assignment $\ell$, the \emph{modularity} is defined as
\begin{equation}
	Q := \frac{1}{2W}\sum_{u,v\in V} \left(A_{uv} - \frac{s_u s_v}{2W}\right) \delta^\ell(u,v) = \frac{1}{2} \sum_{u,v\in V} Q_{uv} \delta^\ell(u,v),
	\label{eq:modularity}
\end{equation}
where 
\begin{equation*}
    \delta^\ell(u,v) =
    \begin{cases}
        1 \quad \text{if} \quad \ell(u) = \ell(v) \\
        0 \quad \text{otherwise}
    \end{cases}
\end{equation*}
and we write
\begin{equation*}
	Q_{uv} := \frac{1}{W}\left(A_{uv} - \frac{s_u s_v}{2W} \right).
\end{equation*}
Note that, like $A$, $Q$ is also symmetric: $Q_{uv} = Q_{vu}$. $Q = Q(\ell)$ in Eq.~\eqref{eq:modularity} should be thought of as a function of $\ell$; however, we will suppress the $\ell$ dependence of $Q$ unless it is ambiguous as to which $\ell$ we are referring. 

For our purposes it will be more convenient to express $Q$ as
\begin{equation}
    Q = \sum_{u<v} Q_{uv} \delta^\ell(u,v) + \frac{1}{2} \sum_{u\in V} Q_{uu}
    \label{eq:modularity2}
\end{equation}
where the second term is a constant independent of the label function $\ell$. Since our objective is to find an $\ell$ that maximizes $Q$, we can safely ignore the second (constant) term in Eq.~(\ref{eq:modularity2}).

\ 

\noindent For a vertex $u \in V$, we call a community $C_{\alpha}$ a \emph{neighboring community} of $u$ if $C_{\alpha} \cap N_u \neq \emptyset$. The Louvain algorithm only moves vertices to neighboring communities. For a vertex $u$ we write
\[
    \zeta_u := \{\alpha \in [n]: C_{\alpha} \cap N_u \neq \emptyset\}
\]
for the set of labels of communities that neighbour $u$, and
\[
    \delta_u = |\zeta_u|
\]
for the number of neighboring communities of $u$. In addition,let
\begin{eqnarray*}
	S^{\alpha}_u := \sum_{v \in C_{\alpha}} A_{uv}, \quad \text{and} \quad
	\Sigma_{\alpha} := \sum_{v \in C_{\alpha}} s_v,
\end{eqnarray*}
i.e. $S^{\alpha}_u$ is the sum of all weights on edges from vertex $u$ to vertices in community $C_{\alpha}$, and $\Sigma_{\alpha}$ is the sum of all weights on edges incident to vertices contained in community $C_{\alpha}$. (Note that in the expression for $S^{\alpha}_u$ we can actually restrict the sum over all $C_{\alpha}$ to $N_{u} \cap C_{\alpha}$, since $A_{uv} = 0$ for all $v\in V$ not neighboring $u$.)

\

The Louvain algorithm attempts to move vertices from one community to the next in a greedy way by only making moves that strictly increase the modularity. Suppose that a vertex $u$ currently in community $C_{\ell(u)}$ is moved to neighboring community $C_{\alpha}$. Then the change in the modularity $\Delta_u^{\alpha}$ resulting from this move is given by
\begin{align}
	\Delta_u^{\alpha} &= \sum_{w \in C_{\alpha}} Q_{uw} - \sum_{w \in C_{\ell(u)} \setminus \{u\}} Q_{uw} \nonumber \\
	&= \frac{1}{W} \sum_{w \in C_{\alpha}} A_{uw} - \frac{s_u}{2W^2}\sum_{w \in C_{\alpha}} s_w  - \frac{1}{W} \sum_{w \in C_{\ell(u)}} A_{uw} + \frac{s_u}{2W^2}\sum_{w \in C_{\ell(u)}\setminus \{u\}} s_w \nonumber \\ 
	&= \frac{S_u^{\alpha} - S_u^{\ell(u)}}{W} - \frac{s_u  \left(\Sigma_{\alpha} - \Sigma_{\ell(u)} + s_u \right)}{2W^2}\,, \label{eq:delta-u-alpha} 
\end{align}
where we have used that $A_{uu} = 0$. Finally, for a fixed vertex $u$, we define $\bar{\Delta}_u := \max_{{\alpha} \in \zeta_{\alpha}} \Delta_u^{\alpha}$. Note that both $\Delta_u^{\alpha} = \Delta_u^{\alpha}(\ell)$ and $\bar{\Delta}_u = \bar{\Delta}_u(\ell)$ depend on $\ell$, but we will again suppress the $\ell$-dependence unless it is relevant for the statement in question.

\subsection{The Louvain algorithm}\label{sec:Louvain}
The Louvain algorithm alternates between two \emph{phases}. The first phase consists of a number of greedy moves that attempt to increase modularity. When there are no more moves left to make, the second phase contracts communities into single vertices, and then the whole process repeats itself at this new coarse-grained level. The two phases repeat until, at some point, no new moves exist directly at the start of a phase 1. 

Because we will introduce several (quantum) versions of the Louvain algorithm, we will refer to the Louvain algorithm from~\cite{blondel08} as the \emph{original Louvain} algorithm (discussed below), or OL for short. For a precise description of the algorithm, see Algorithm~\ref{alg:C-Louvain}.

\begin{algorithm}[!htb]
  \caption{The Louvain algorithm}
  \label{alg:C-Louvain}
   \begin{algorithmic}[1]
    \Function{Louvain}{Graph $G$, Community set $\mathcal{C}$}
    \State  $\mathcal{C} \gets$ \Call{SinglePartition}{G}
    \Comment assign each node its own community
    \State done $\gets$ False
    \While{not done}
        \State  $\mathcal{C'} \gets$ \Call{MoveNodes}{G, $\mathcal{C}$}
        \Comment get new community assignment
        \State done   $\gets \|\mathcal{C}\| = \|V\|$
        \Comment end when every community consists of one node
        \If{not done}
            \State G $\gets$ \Call{AggregateGraph}{G, $\mathcal{C}$} 
            \State  $\mathcal{C} \gets$ \Call{SinglePartition}{G}
        \EndIf
    \EndWhile
    \State \Return $\mathcal{C}$
    \EndFunction
    \end{algorithmic}


\begin{algorithmic}[1]
    \Function{MoveNodes}{Graph G, Community set $\mathcal{C}$}
        \State done $\gets$ False
        \While{not done}
            \State done $\gets$ True
            \ForAll{$u \in V$}
                \State $\bar{\Delta}_u \gets $ $\max_{v \in N_u}\Delta_{uv}$
                \Comment calculate maximum increase of modularity
                \If{$\bar{\Delta}_u > 0$}
                    \State $\bar{v} \gets $ $\argmax_{v \in N_u} \Delta_{uv}$
                    \Comment get corresponding community 
                    \State $\ell(u) \gets \ell(\bar{v})$
                    \Comment reassign $u$ to community of $\bar{v}$
                    \State done $\gets$ False 
                    \Comment terminate when there is no modularity increase
                \EndIf
            \EndFor
        \EndWhile
        
    \EndFunction
\end{algorithmic}


    \begin{algorithmic}[1]
    \Function{AggregateGraph}{Graph G, Community set $\mathcal{C}$}
        \State $V' \gets \{C_a \| C_a \neq \emptyset\}$
        \Comment create new vertex for every nonempty set
        \State $A' \gets \{A'_{ab} |a,b \in V', A'_{ab} = \sum_{u \in C_a, v \in C_b} A_{uv} \}$
        \State $E' \gets \{ (a,b) | a,b \in V', A_{ab}' > 0\}$
        \State \Comment create edges with weight equal to the sum of all weights between vertices in each community 
        \State \Return \Call{Graph}{$V'$, $E'$, $A'$}
    \EndFunction
\end{algorithmic}


\begin{algorithmic}[1]
    \Function{SinglePartition}{Graph G}
        \State  \Return $\{ \{v\}| v\in V\}$
    \EndFunction
   \end{algorithmic}
\end{algorithm}

\subsubsection*{Initialization}
\label{sec:OL_init}

\noindent Initially every vertex is assigned to its own community $\ell(u) = u$. Before beginning, for every $u \in V$, loop over all neighbors $j \in [d_u]$ in order to compute the vertex strengths $s_u$ as well as each $\Sigma_{\ell(u)} = s_u$, and also the total edge weight sum $W = \frac{1}{2} \sum_{uv} A_{uv}$. If not already sorted, we  also sort all adjacency lists during initialization.

\subsubsection*{First phase}
\label{sec:OL_first_phase}

\noindent During the first phase, the algorithm places all vertices in a randomly ordered list. This list is traversed sequentially and, for each vertex encountered, we compute $\bar{\Delta}_u = \max_{\alpha\in \zeta_u} \Delta_u^{\alpha}$. If $\bar{\Delta}_u>0$, $u$ is moved to the community that realises $\argmax_{\alpha \in \zeta_a} \Delta_u^{\alpha}$. After completing a pass through the list, it is reshuffled and the process is repeated. This phase ends when there are no vertices that can be moved to increase the modularity any further, i.e. when $\bar{\Delta}_u \leq 0$ for all $u\in V$.

In order to compute $\bar{\Delta}_u$ for a given a $u\in V$, the algorithm can first construct the list of neighboring community labels $\zeta_u$ as well as a list $L_u = \{(\alpha, S_u^{\alpha}): \alpha \in \zeta_u\}$ of neighboring communities and corresponding sums of edge weights from $u$ to those communities. Now a single loop over $L_u$ is sufficient to compute $\Delta_u^{\alpha}$ for every $(\alpha, S^{\alpha}_u) \in L_u$ and output $\bar{\Delta}_u$ and $\bar{\alpha} = \argmax_{\alpha \in \zeta_u} \Delta_u^{\alpha}$. If $\bar{\Delta}_u > 0$, then $u$ is moved from its original community to the new community $C_{\bar{\alpha}}$. 

As vertices move from one community to the next, the algorithm maintains a list of the sums $\{\Sigma_{\alpha}: \alpha \in [n]\}$. In particular, after moving vertex $u$ from its original community $C_{\beta}$ to its new community $C_{\bar{\alpha}}$, we subtract $s_u$ from $\Sigma_{\beta}$ and add it to $\Sigma_{\bar{\alpha}}$ to ensure that the quantities $\{\Sigma_{\alpha}: \alpha \in [n]\}$ are kept up to date. The algorithm then also updates the label function $\ell$.

\subsubsection*{Second phase}
\label{sec:OL_second_phase}

After the first phase has finished and no vertex move can further increase the modularity, a new coarse-grained graph $G' = (V', E')$ with weighted adjacency matrix $A'$ is constructed. This new coarse-grained graph has as its vertex set the set of (non-empty) communities constructed in the first phase: $V' = \{C_{\alpha}: \alpha \in [n], C_{\alpha} \neq \emptyset\}$. An edge is present in $G'$ between two vertices $C_{\alpha}, C_{\beta} \in V'$ if there is an edge between any two vertices in corresponding communities in $G$, i.e. $E' = \{(C_{\alpha},C_{\beta}): \exists u\in C_{\alpha}, v \in C_{\beta} \text{ such that } (u,v)\in E\}$, and its weight is the sum of all the edge weights of edges between $C_{\alpha}$ and $C_{\beta}$ in $G$, i.e. $A_{\alpha \beta}' = \sum_{u\in C_{\alpha}, v\in C_{\beta}} A_{uv}$. Constructing $V'$, $E'$ and $A'$ can be done in a single loop over all edges of $G$.

\subsection{Complexity of the Louvain algorithm}
\label{sec:classical}
As discussed, the goal of this paper is to compare the performance of the Louvain algorithm to its quantum counterparts using the empirical method outlined in the introduction, and in doing so to tackle some of the obstacles and considerations that one might face in taking such an approach. As we discussed there, this means that we must choose a measure of complexity for our algorithms to use to make our (empirical) comparisons. Concretely, we choose to count the number of calls to the function that computes the change in modularity resulting from moving a vertex from one community to another. In this section we first precisely define our complexity measure, and then analyse the complexity of the classical Louvain algorithm.

\subsubsection{Complexity measure}
\label{sec:runtime_comparison}

We consider how many calls are made to the function that computes the change in modularity resulting from a particular vertex move. I.e. for a particular vertex $u$ and community $\alpha$ that it might move to, we count calls to (an oracle that computes) the function\footnote{When we perform our numerical study in Section~\ref{part:three}, we will count calls to the function $g_{\Delta}$, and each oracle call will correspond to $c_q = 2$ function calls.}
\begin{equation}
    g_\Delta(s_u, \Sigma_\alpha, \Sigma_{\ell(u)}, S_u^\alpha, S_u^{\ell(u)}) = \frac{S_u^\alpha - S_u^{\ell(u)}}{W} - \frac{s_u\left( \Sigma_\alpha - \Sigma_{\ell(u)} - s_u\right)}{2W^2}\,.
\end{equation}
Note that this can be seen counting the number calls to (the gradient of) the modularity $Q$ that we are attempting to maximize, which is often the natural measure of complexity for optimization algorithms.

Counting the number of function calls does not capture every part of the algorithms' complexities. Recall that the classical Louvain algorithm consist of several phases, in which the initialization phase does not require any function calls and the second (coarse-graining) phase similarly does not, even though both require a single loop over all $m$ edges. By taking the number of function calls as a means of comparing the quantum and classical algorithms, we are inherently not taking into account the initialization and second phases. However, in practice the first phase takes up the vast majority of the computation time, and moreover the initialization and second phases are identical for both the classical and quantum algorithms, and hence ignoring them in our comparisons is sensible. Another aspect that is not measured by the number of function calls is the time it takes to compute the list $L_u$ for each vertex $u$ considered. In Section~\ref{sec:complexity_of_louvain} below, we argue that the classical algorithm can be improved upon slightly by keeping all of these lists in memory, and only updating those that change after moving a vertex. The time it takes to update this list is also not captured in the comparison between the quantum and classical algorithms, however, as with the initialization and second phases, these updates involve the same operations for both the quantum and classical algorithms.\footnote{We would have liked to include the time it takes to update the lists also in our analysis. However, working out the precise complexity of doing so is impossible without assuming a particular quantum hardware architecture. In particular, it necessitates the introduction of several new architecture-dependent variables that encompass how read and write times compare between classical and quantum memories, and how these times compare to the cost of computing the function $g_{\Delta}$. Including these quantities as tune-able parameters in the analysis would have made things less clear. By choosing to count only the number of function calls as a means of comparison, we have chosen not to focus on the part of the algorithm that involves updating data structures.}

\

Instead of using the number of function calls to the modularity function, a common complexity measure for graph algorithms is the number of queries to the graph. However, the initialization and second phases of the algorithm already require us to query all edges of the input graph. Beyond this no further queries need to be made, since we can simply query all edges once and then store the results in memory, and thus all of our algorithms have `query complexity' $|E|$. Hence, to meaningfully compare the algorithms we consider in this work empirically, we will simply count the number of calls to the function $g_\Delta$.

\

In the sections that follow we will also consider the \textit{time/gate complexity} of the algorithms, in terms of how many additional elementary operations are required besides the calls to $g_\Delta$. This will be useful in order to compare the worst-case asymptotic run-times of the classical and quantum algorithms, but for the purposes of numerical comparison it is much cleaner to focus only on calls to $g_\Delta$, since the precise number (i.e. including constants) of elementary gates required for various operations quickly becomes architecture-dependent.

\subsubsection{Complexity of Louvain}
\label{sec:complexity_of_louvain} 

For every vertex $u$ visited during the steps of the first phase, there are $\delta_u$ calls to $g_\Delta$ (one for each community $\alpha$ adjacent to $u$) required to compute $\bar{\Delta}_u$ and $\bar{\alpha}$, as well as $O(d_u)$ other operations needed to construct the list $L_u$. The total complexity then depends on how many vertices are visited in the entire first phase of the algorithm. If we suppose the algorithm makes $T$ moves in total, and that on the $k$th move it must inspect $t_k$ vertices before finding one that it can move, then the total number of function calls (queries) required by the algorithm will be
\[
    \sum_{k \in [T]} O\left( \delta_{\max}t_k \right)\,,
\]
and the number of other operations 
\[
    \sum_{k \in [T]} O\left( (\delta_{\max}+d_{\max})t_k \right)\,.
\]
Since the algorithm is heuristic, it is difficult to accurately bound the total number of moves $T$. In Appendix~\ref{app:louvain_total_moves} we show that, in general, $T$ can be upper bounded by a polynomial in $n$, and hence the Louvain algorithm is always a polynomial-time algorithm. In practice, however, $T$ often scales as $O(n \log n)$~\cite{lancichinetti2009community} -- as confirmed also by our numerical data presented in Appendix~\ref{sec:number_moves}.

\

\noindent The run-time (but not the number of function calls) of the original classical algorithm can be improved slightly at the expense of a constant overhead in space complexity, by making use of an additional data structure. Recall that the classical algorithm computes, for each vertex $u$ that it visits, a list $L_u = \{(\alpha, S_u^\alpha) : \alpha \in \zeta_u\}$, and then uses these values (plus the $\Sigma_{\alpha}$'s and $s_u$'s also stored in memory) as input to the function $g_\Delta$. This list takes $\tilde{O}(d_u)$ time to construct (since we must loop over all neighbours of $u$ to compute the appropriate sums), leading to $\tilde{O}((\delta_{\max}+d_{\max})t_k)$ time required for the $k$th step. We can improve this complexity if we store the information contained in $L_u$ for every $u$ separately, and update it as appropriate.

\paragraph{Data structure} In particular, for each $u\in V$ we introduce a `community adjacency list' $\eta_u : [\delta_u] \rightarrow [n] \times \mathbb{R}$, which given an index $j$, returns the label $\alpha$ of the $j$th neighbouring community to $u$, as well as the sum $S_u^{\alpha}$. As shorthand we will often write $\alpha = \eta(j)$, even though $\eta(j)$ actually returns the tuple $(\alpha, S_u^{\alpha})$.
We will keep the list $\eta_u$ sorted by community label (according to some arbitrary but fixed ordering), allowing lookup of $S_u^\alpha$ using label $\alpha$ in $O(\log \delta_u)$ time. Finally, we will reserve a special place in this list to store the quantity $S_u^{\ell(u)}$, and assume that we can access this directly. 
We will refer to the sets $\{s_u: u\in V\}$, $\{\Sigma_{\alpha}: \alpha \in [n]\}$ and $\{\eta_u : u\in V\}$ collectively as the \emph{data structure}.

The data structure therefore allows us to obtain the inputs to $g_\Delta$ all in constant time. Now we concern ourselves with the time required to update it. Suppose that we move vertex $u$ from community $\alpha$ to community $\beta$. In terms of the $\Sigma_{\cdot}$ values, it is clear that only $\Sigma_\alpha$ and $\Sigma_\beta$ change. These are easily updated by subtracting $s_u$ from the former and adding it to the latter, which requires $O(\log n)$ time. 
The only community adjacency lists $\eta_v$ that will change will be for vertices $v$ that are neighbours of $u$: since each entry in any $\eta_v$ stores only sums of weights of edges incident to neighbouring communities of $v$, any sum that doesn't include an edge to $u$ will remain unchanged. Within each $\eta_v$ (for $v \in N_u$), the only sums that will change will be the ones corresponding to the communities that have changed: namely, $S_{v}^\alpha$ and $S_{v}^\beta$. The list $\eta_v$ is sorted by community label, and so we can identify the indices $i$ and $j$ corresponding to communities $\alpha$ and $\beta$ in $O(\log d_{\max})$ time each using binary search. Then we update the tuple $\eta_v(i) = (\alpha, S_{v}^\alpha)$ by subtracting $A_{uv}$ from $S_{v}^\alpha$, and we update the tuple $\eta_v(j) = (\beta, S_{v}^\beta)$ by adding $A_{uv}$ to $S_{v}^\beta$, where each operation will take time $O(\log d_{\max})$. If we find that the new value of $S_{v}^\alpha$ is equal to zero, we remove that tuple from the list, and if the tuple $(\beta, S_{v}^\beta)$ does not already exist, then we create it and insert it into the list at its sorted position. Note that we will only add a new tuple if the list is not already of length $d_v$, and hence the length of the list remains less than or equal to $d_v$. 

For each neighbour $v$ of $u$, these updates therefore take time $O(\log d_{\max})$. Since we do this for every neighbour of $u$, the total time for all updates is $O(d_u \log _{\max}) \leq O(d_{\max} \log d_{\max})$.

\paragraph{Complexity with the data structure} By using the data structure described above, we can eliminate the need to construct the list $L_u$ for each vertex $u$, at the cost of having to update the data structure after every move. In this case, the number of function calls remains the same, but the classical algorithm now takes time 
\[
    \sum_{k=1}^T \tilde{O}\left( \delta_{\max}t_k + d_{\max}\right)\,.
\]
We verify numerically in Appendix~\ref{sec:runtime_OL_LL} that the addition of the data structure does indeed improve the run-time of the algorithm.

\section{Quantum algorithms for community detection}
\label{sec:quantum}
In this section we present a number of quantum variants of the original Louvain algorithm for community detection. Our reason for introducing several quantum algorithms is to later study, in Section~\ref{part:three}, how much of the promised asymptotic (per-step) speedup actually manifests in practice for different variants of the algorithm, and to demonstrate how an empirical comparison between algorithms can reveal significant differences in their run-times that aren't made clear by an asymptotic analysis alone. However, in this section we will only concern ourselves with the usual kind of asymptotic analysis of algorithm complexity.

We begin by introducing a quantum algorithm that mimics the classical algorithm exactly (i.e. by searching for the first good vertex from a randomly ordered list of vertices), and which makes use of a nested Grover search. We then construct a much-simplified variant that forgoes the ordered list and directly applies a nested Grover search to the entire set of vertices. Both of these algorithms also make use of quantum maximum finding to obtain the \textit{best} move available to a particular good vertex. In the end these algorithms are somewhat sub-optimal: the nested Grover searches are performed over sets of varying sizes, but the outer Grover search complexity is limited by the size of the largest set, something that does not happen in the classical case. To overcome this drawback, in Section~\ref{sec:QVTAA} we introduce a slightly more sophisticated quantum algorithm that makes use of the technique of \textit{variable time amplitude amplification}, which allows the subroutine called by a Grover search to have different stopping times. However, since we do not numerically study this algorithm, we defer its details to Appendix~\ref{app:VTAA}. Later in the section we consider dropping the nested Grover search format all together in favour of an asymptotically sub-optimal, but likely practically more efficient, implementation that instead makes use of a classical subroutine, and which might be much more efficient on sparse input graphs. Finally, we present a much-simplified quantum algorithm that performs a single Grover search over the space of \emph{edges} of the graph, in search of one that suggests a good move.

\subsection{Quantum preliminaries}
We will find the following quantum subroutines useful. Later, in Section~\ref{part:three}, we will consider explicit implementations of them as given in~\cite{ourotherpaper}, and take into account their full run-times (i.e. including all constants). For \textbf{QSearch}  (Lemma~\ref{lem:grover} below), this will in particular mean including an extra argument ($N_{\text{samples}}$) to the algorithm that determines how many classical samples are drawn before Grover search is used, but which does not affect the asymptotic runtime. 

\begin{lemma}[Grover's search with an unknown number of marked items~\cite{boyer1998tight}]
\label{lem:grover}
Let $L$ be a list of items, and $t$ the (unknown) number of `marked items'. Let $\mathcal{O}_g \ket{x_i}\ket{0} = \ket{x_i}\ket{g(x_i)}$ be an oracle that provides access to the Boolean function $g : [|L|] \rightarrow \{0,1\}$ that labels the items in the list. Then there exists a quantum algorithm \textbf{QSearch}$(L,\epsilon)$ that finds and returns an index $i$ such that $g(x_i) = 1$ with probability at least $1-\epsilon$ if one exists and requires an expected number $O(\sqrt{N/t} \log(1/\epsilon))$ queries to $\mathcal{O}_g$ and $O(\sqrt{N/t}\log(N/\epsilon))$ other elementary operations. If no such $x_i$ exists, the algorithm confirms this and to do so requires $O(\sqrt{N}\log(1/\epsilon))$ queries to $\mathcal{O}_g$ and $O(\sqrt{N}\log(N/\epsilon))$ other elementary operations. 
\end{lemma}

\begin{lemma}[Exact Grover search~\cite{hoyer2000arbitrary}]\label{lem:exact_grover}
    Let $L$ be a list of items, and $t>0$ the \emph{known} number of `marked items'. Let $\mathcal{O}_g \ket{x_i}\ket{0} = \ket{x_i}\ket{g(x_i)}$ be an oracle that provides access to the Boolean function $g : [|L|] \rightarrow \{0,1\}$ that labels the items in the list. Then there exists a quantum algorithm \textbf{ExactQSearch}$(L,t)$ that finds and returns an index $i$ such that $g(x_i) = 1$ with \emph{certainty}. To do so, the algorithm makes $O(\sqrt{N/t})$ queries to $\mathcal{O}_g$ and $O(\sqrt{N}\log(N))$ other elementary operations. 
\end{lemma}

\begin{lemma}[Quantum maximum-finding \cite{durr1996quantum}]
\label{lem:quantum_max}
Let $L$ be a list of items of length $|L|$, with each item in the list taking a value in the interval $[a,b]$, to which we have coherent access in the form of a unitary that acts on basis states as 
	\[
		\mathcal{O}_L \ket{x}\ket{0} = \ket{x}\ket{L[x]}.
	\]
Then there exists a quantum algorithm \textbf{QMax}$(L,\epsilon)$ that will return $\argmax_x L[x]$ with probability at least $1-\epsilon$ using at most $O(\sqrt{|L|}\log(1/\epsilon)$ queries to $\mathcal{O}_f$ (i.e. to the list $L$) and $O(\sqrt{|L|}\log |L|\log(1/\epsilon))$ elementary operations.
\end{lemma}

~\\
We will also make use of the variable time amplitude amplification (VTAA) algorithm of Ambainis~\cite{ambainis12}. The statement of this result is somewhat more involved, and so we will defer to Appendix~\ref{app:VTAA} for a more formal description of VTAA and its run time in the context of our particular application of it, and discuss the technique informally here.

Consider a quantum algorithm $\mA$ which may stop at one of several times $t_1,\dots,t_m$. To indicate the outcome, $\mA$ has an extra register $O$ with $3$ possible values $0$, $1$, and $2$: $0$ indicates that the computation has stopped but did not reach the desired outcome; $1$ indicates that the computation has stopped and the desired outcome was reached; $2$ indicates that the computation has not stopped yet. The idea behind VTAA is to run multiple branches of computation in superposition, and to amplify those branches that have either stopped and reached the desired outcome ($1$) (e.g. found a marked item), or are still running ($2$).

Let $p_i$ be the probability of the algorithm stopping at time $t_i$ (with either the outcome $0$ or outcome $1$). The average stopping time of $\mA$ (the $l_2$ average) is
\begin{equation}
    T_{\text{avg}} := \sqrt{\sum_i p_i t_i^2}.
    \label{eq:stopping_times}
\end{equation}
Let $T_{\text{max}} = t_m$ be the maximum possible running time of $\mA$, 
\[
    \alpha_{\text{good}}\ket{1}_O\ket{\psi_{\text{good}}} + \alpha_{\text{bad}}\ket{0}_O\ket{\psi_{\text{bad}}}
\]
be the final state of the algorithm once all branches have stopped, and $p_{\text{succ}} = |\alpha_{\text{good}}|^2$ be the probability of obtaining the state $\ket{\psi_{\text{good}}}$ using algorithm $\mA$. Then Ambainis~\cite{ambainis12} shows the following.
\begin{lemma}[Variable time amplitude amplification~\cite{ambainis12}]
\label{lem:VTAA}
    There exists a quantum algorithm $\mA'$ invoking $\mA$ several times, for total time
    \[
        \tilde{O}\left(T_{\max}\log(T_{\max}) + \frac{T_{\text{avg}}}{\sqrt{p_{\text{succ}}}} \log^{1.5} T_{\max} \right)
    \]
that produces a state $\alpha\ket{1}\ket{\psi_{\text{good}}} + \beta\ket{0}\ket{\psi'}$ such that $|\alpha|^2 > 1/2$.  By repeating $\mA'$ $O(\log \frac{1}{\epsilon})$ times, we can obtain $\ket{\psi_{\text{good}}}$ with probability at least $1-\epsilon$.
\end{lemma}
This is in contrast to the usual amplitude amplification routine, which would take time $O(T_{\max}/\sqrt{p_\text{succ}})$, and hence we see a speedup whenever $T_{\text{avg}}$ is substantially smaller than $T_{\max}$. However, the algorithm $\mA$ must satisfy a number of constraints (in particular, it cannot be adaptive), and so VTAA is not always applicable. This will become clear when we describe our algorithm in Section~\ref{sec:QVTAA}.

\

Finally, we will assume that we have access to quantum read/classical write RAM (QRAM), where a single QRAM operation is considered to be classically writing a bit to the QRAM or making a quantum query (a read operation) to bits stored in QRAM, possibly in superposition. See~\cite{apers2020quantum} for a more detailed discussion. 

\subsection{Quantum community detection}
In the sections that follow we describe our various quantum algorithms for community detection. These algorithms are (roughly in order of increasing simplicity):
\begin{itemize}
    \item \textbf{QLouvain} -- A quantum version of classical Louvain (Section~\ref{sec:quantum_louvain}).
    \item \textbf{SimpleQLouvain} -- A much simplified version of \textbf{QLouvain} that deviates slightly from the behaviour of the original Louvain algorithm (Section~\ref{sec:grover_louvain}).
    \item \textbf{QLouvainSG} and \textbf{SimpleQLouvainSG} -- Versions of both algorithms above that are more efficient if the input graph is sparse (Section~\ref{sec:sparse_graph_versions}).
    \item \textbf{EdgeQLouvain} and \textbf{NodeComQLouvain} -- Two vastly simplified algorithms that deviate substantially from the spirit of the original Louvain algorithm, but nevertheless obtain similar results in practice (Section~\ref{sec:EdgeLouvain}).
\end{itemize}
We also describe an approach based on variable time amplitude amplification in Section~\ref{sec:QVTAA} that yields (asymptotically) more efficient versions of the first four algorithms above. However, these algorithms are much more complicated than those described above, and therefore we have chosen not to simulate these numerically in Section~\ref{sec:numerics}. 

\subsubsection{Quantum louvain}\label{sec:quantum_louvain}
Our first quantum algorithm works by identifying the first vertex in a list for which there exists a good move, and then moves it, just as the classical algorithm does. Using the quantum algorithm \textbf{FindFirst} (introduced below) we obtain in this way a square-root improvement over the per-step classical complexity.

We begin by describing a quantum algorithm that performs a quantum search over a list of vertices in order to identify one for which a good move exists. The algorithm comes with a bound on the \emph{expected} run-time -- which benefits from having more good moves and good vertices available -- and a bound on the \emph{worst-case} run-time, which forgoes the aforementioned benefits. We will use the latter bound in our analysis of the main algorithm, since it is insensitive to the number of marked items, but in fact the run-time would be improved in practice by taking into account the actual number of good vertices.

\begin{lemma}\label{lem:vertex_find}
    There exists a quantum algorithm {\bf VertexFind}($L,\zeta$), which, given a list $L$ of vertices $u_0,\dots,u_{|L|-1}$, returns the identity $i$ of a vertex $u_i$ such that $\bar{\Delta}_{u_i} > 0$ (i.e. a good vertex) with probability $\geq 1-\zeta$ if one exists, and otherwise returns `no vertex exists'. The algorithm requires an expected number of function calls at most
    \[
        O\left(\sqrt{\frac{\delta_{\max}}{f}} \log\left(\frac{{|L|}}{\zeta}\right) \right) = \tilde{O}\left(\sqrt{\frac{\delta_{\max}}{f}} \log\left(\frac{1}{\zeta}\right) \right) \,,
    \]
    and 
    \[
        O\left(\sqrt{\frac{\delta_{\max}}{f}} \log(|L|) \log(\delta_{\max})  \log\left(\frac{{|L|}}{\zeta}\right) \right) = \tilde{O}\left(\sqrt{\frac{\delta_{\max}}{f}} \log\left(\frac{1}{\zeta}\right) \right) \,
    \]
    elementary operations, where $f$ is the fraction of vertices in $L$ that are good (and the $\tilde{O}$ notation hides polylogarithmic factors in $|L|$ and $\delta_{\max}$). If we want to obtain a \emph{worst case} run-time, then there is a variant of the algorithm that behaves the same, but requires in the worst case at most 
    \[
        O\left(\sqrt{\delta_{\max}|L|} \log\left(\frac{{|L|}}{\zeta}\right) \right) = 
        \tilde{O}\left(\sqrt{\delta_{\max}|L|} \log\left(\frac{1}{\zeta}\right) \right)\,
    \]
    function calls and 
    \[
        O\left(\sqrt{\delta_{\max}|L|} \log(|L|)\log(\delta_{\max}) \log\left(\frac{{|L|}}{\zeta}\right) \right) = 
        \tilde{O}\left(\sqrt{\delta_{\max}|L|} \log\left(\frac{1}{\zeta}\right) \right)\,
    \]
    elementary operations.
\end{lemma}
\begin{proof}
\noindent We apply Grover search to find, for a particular vertex $u$, an integer $j \in [\delta_u]$ such that $\Delta_u^{\eta_u(j)} > 0$ (a `good move'), if one exists. Using this as a subroutine, we apply Grover search now to the list $L$ to find any vertex for which there exists such a neighbouring community (a `good vertex').

Using the data structure described in Section~\ref{sec:complexity_of_louvain}, we can obtain the inputs to $g_\Delta$, which computes the change in modularity resulting from moving a vertex $u$ to a community $\alpha$, in constant time: we can recover from $\eta_u$ the quantity $S_u^\alpha$ (this is just the weight associated to the entry $\eta_u(j)$), and also obtain $S_u^{\ell(u)}$, $s_u$, $\Sigma_a$, and $\Sigma_{\ell(u)}$ directly in $O(1)$ time from the appropriate lists. We will use $\mathcal{A}_{g_\Delta}$ to denote the unitary that implements the (classical) sub-routine for computing $g_\Delta(s_u, \Sigma_\alpha, \Sigma_{\ell(u)}, S_u^\alpha, S_u^{\ell(u)}) =: \Delta_{u}^{\eta_u(j)}$ given $u$ and $j$, and whose action on basis states is 
\[
    \ket{u}\ket{j}\ket{0} \mapsto \ket{u}\ket{j}\ket{\Delta_{u}^{\eta_u(j)}}\, . 
\]
For a fixed vertex $u$, we can find a $j$ such that $\Delta_u^{\eta_u(j)} > 0$ if one exists using the algorithm {\bf QSearch} of Lemma~\ref{lem:grover}, by providing $\mathcal{A}_{g_\Delta}$ as an oracle. The algorithm will require in the worst case $O(\sqrt{\delta_u}\log(1/\epsilon))$ uses of $\mathcal{A}_{g_\Delta}$ (and hence $g_\Delta$) and its inverse, and $O(\sqrt{\delta_u}\log |L| \log(1/\epsilon))$ other operations to find one with probability at least $1-\epsilon$, or to signal that no such $j$ exists as appropriate. We will write $\mathcal{A}_{\bar{\Delta}}$ to denote the unitary that implements this quantum algorithm for a given vertex $u$, i.e. it maps
\[
    \ket{u}\ket{0} \mapsto \ket{u}\ket{\bar{\Delta}_u > 0 ?}
\]
where the last register on the right is 1 if $\bar{\Delta}_u = \max_{j} \Delta_u^{\eta_u(j)} > 0$, and $0$ otherwise.

We then use $\mathcal{A}_{\bar{\Delta}}$ as a subroutine to search for a vertex $u$ such that $\bar{\Delta}_u > 0$, i.e. for which there exists a $j$ such that $\Delta_u^{\eta(j)} > 0$ (a good vertex). This can be achieved via another straightforward application of {\bf QSearch} from Lemma~\ref{lem:grover}, with $\mathcal{A}_{\bar{\Delta}}$ provided as the oracle. If the probability of a randomly chosen vertex $u$ having a good move is $\frac{1}{f}$, $f>0$, then this will require an expected $O(\sqrt{1/f})$ applications of $\mathcal{A}_{\bar{\Delta}}$ and its inverse, and when $f=0$, the worst case, it will require at most $O(\sqrt{|L|})$ applications to signal that no there are no good vertices. Assuming that the sub-routine $\mathcal{A}_{\bar{\Delta}}$ works perfectly, we can boost the success probability of the algorithm from $2/3$ to $1-\epsilon'$ by repeating $O(\log(1/\epsilon'))$ times. However, the subroutine $\mathcal{A}_{\bar{\Delta}}$ succeeds only with probability $\geq 1-\epsilon$. Hence for the outer search algorithm to work correctly we will need that every time the $\mathcal{A}_{\bar{\Delta}}$ subroutine is run, it succeeds. Since $\mathcal{A}_{\bar{\Delta}}$ (and its inverse) will be called at most $O(\sqrt{|L|}\log(1/\epsilon'))$ times, the entire search algorithm will therefore succeed with probability at least
\[
	(1-\epsilon') \cdot (1-\epsilon)^{O(\sqrt{|L|}\log(1/\epsilon'))}.
\]
To ensure that this probability is $\geq 1-\zeta$, we can choose $1/\epsilon' = O(\poly(1/\zeta))$ and $1/\epsilon = \poly(|L|,1/\epsilon')$. In particular, we can set $\epsilon' = \zeta/2$ and $\epsilon = \frac{\zeta/2}{ \sqrt{|L|}\log(1/\epsilon')}$. 

Finally, since the algorithm $\mathcal{A}_{\bar{\Delta}}$ is called in superposition for multiple vertices $u$, the run-time of the outer {\bf QSearch} routine will be limited by its slowest branch\footnote{Note that the routine will really be limited by a known \textit{upper bound} on the time taken by any particular branch. For us, this will be $O(\sqrt{\delta_{\max}} \log(1/\epsilon))$, but this does require us to know $\delta_{\max}$. Luckily, we can keep track of this over time by adding to our data structure, and the overheads for updating it will all at worst be logarithmic in $n$ and linear in $d_{\max}$, similarly to the other updates.}, which requires at most $O(\sqrt{\delta_{\max}} \log(1/\epsilon))$ queries to $g_\Delta$, plus $O(\sqrt{\delta_{\max}}\log (\delta_{\max}) \log(1/\epsilon))$ other operations, meaning that the total expected number of function calls made by the entire algorithm is at most 
\[
    O\left(\sqrt{\frac{\delta_{\max}}{f}} \log\left(\frac{{|L|}}{\zeta}\right) \right)\,,
\]
and the total expected number of other operations is 
\[
    O\left(\sqrt{\frac{\delta_{\max}}{f}} \log(|L|) \log(\delta_{\max})  \log\left(\frac{{|L|}}{\zeta}\right) \right)\,.
\]
Using the worst-case upper bound of {\bf QSearch} from Lemma~\ref{lem:grover}, we can bound the total \emph{worst-case} number of function calls by
\[
    O\left(\sqrt{\delta_{\max}|L|} \log\left(\frac{{|L|}}{\zeta}\right) \right)\,
\]
and the total number of other operations by 
\[
    O\left(\sqrt{\delta_{\max}|L|} \log(|L|)\log(\delta_{\max}) \log\left(\frac{|L|}{\zeta}\right) \right)\,.
\]
\end{proof}

\ 

We will use the algorithm \textbf{VertexFind} as a subroutine to implement the quantum algorithm that searches an ordered list of vertices for the first vertex with a good move available. We can now describe this algorithm in detail.

\begin{lemma}\label{lem:find_first}
    Given an ordered list $L$ of vertices $u_0,\dots,u_{|L|-1}$, there exists a quantum algorithm \textbf{FindFirst}($L,\epsilon$), which, with probability $\geq (1 - \epsilon)$, returns $i = \min_j \{j : \bar{\Delta}_{u_j} > 0 \}$, i.e.~the index of the first good vertex in the list if such a vertex exists, and otherwise returns `no good vertex exists'. The algorithm requires at most 
    \[
        O\left( \sqrt{\delta_{\max}i} \log\left(\frac{|L|}{\epsilon}\right) \right)
    \]
    function calls and $\tilde{O}\left( \sqrt{\delta_{\max}i} \log\left({|L|}/{\epsilon}\right) \right)$ other elementary operations in the case that there does exist a good vertex, and otherwise requires at most 
    \[
        O\left( \sqrt{\delta_{\max}|L|} \log\left(\frac{|L|}{\epsilon}\right) \right)
    \]
    function calls and $\tilde{O}\left( \sqrt{\delta_{\max}|L|} \log\left({|L|}/{\epsilon}\right) \right)$ other elementary operations. 
\end{lemma}
\begin{proof}
    Let $i$ be the index of the first good vertex in $L$, and let $q$ be such that $2^q$ is the smallest power of $2$ larger than $i$, and for clarity assume that the length of $L$ is a power of $2$ (this is without loss of generality -- we can always pad $L$ and incur at most a constant overhead in the run-time). The algorithm will proceed in two stages: first, we identify the segment of $L$ in which $i$ lies; then once we have identified the segment, we perform a binary search to identify the precise location of $i$ in that segment. 
    
    Using the algorithm \textbf{VertexFind} from Lemma~\ref{lem:vertex_find}, we search over regions of $L$ that double in size each time, in order to identify an upper bound $j$ on $i$ satisfying $i \leq j \leq 2^q$. In particular we repeat the following routine, initialising $l = 0$ and $r = 1$
    \begin{enumerate}
        \item Let $J = u_{l},\dots,u_{r}$ be the sub-list of vertices in $L$ between $l$ and $r$. Run \textbf{VertexFind}($J,\zeta$) to find a good vertex in $J$, or to determine (with probability $\geq 1-\zeta$) that none exists. This requires at most $\tilde{O}(\sqrt{\delta_{\max}|J|}\log(1/\zeta))$ function calls and other elementary operations.
        \item If a good vertex was found at index $j$, then return $l$ and $j$ and stop.
        \item Otherwise set $l \gets r+1$ and $r \gets 2r$. If $l > |L|$, return `no good vertex exists' and stop; otherwise go to step 1.
    \end{enumerate}
    If the above routine fails to find a good vertex, then we can simply output `no good vertex exists' and stop. If instead there is a good vertex at position $i$, then with high probability (specifically $\geq (1-\zeta)^q$) we will detect this, and output an index $j$ that we know to be an upper bound to $i$. It is only an upper bound since one segment of $L$ might contain multiple good vertices and \textbf{VertexFind} will return any one of these, and so all we learn is that the vertex we're looking for is either at that position or before it. Similarly, since \textbf{VertexFind} will have failed to find any good vertex in any preceding segment that doesn't include $i$, we know a lower bound on the position of $i$, namely $l$. 
    
    From Lemma~\ref{lem:vertex_find}, the number of function calls made by \textbf{VertexFind} on a list of size $a$ with failure probability $\zeta$ is $O\left(\sqrt{\delta_{\max}a} \log\left(\frac{a}{\zeta}\right) \right)$, and hence the run-time of the above procedure to find the segment of $L$ containing $i$ is 
    \begin{eqnarray*}
        \sum_{k=0}^{q} O\left(\sqrt{\delta_{\max}2^k} \log\left(\frac{2^k}{\zeta}\right) \right) &\leq& \sum_{k=0}^{q} O\left(\sqrt{\delta_{\max}2^k} \log\left(\frac{i}{\zeta}\right) \right) \\
        &=& O\left(\sqrt{\delta_{\max}2^q} \log\left(\frac{i}{\zeta}\right) \right) \\
        &\leq& O(\sqrt{\delta_{\max}{i}}\log(i/\zeta))\,,
    \end{eqnarray*}
    where the first and last inequalities follow since $q = \lceil \log_2 i \rceil$ (and hence $2^q \leq 2i$). 
    
    \ 
    
    Once we have lower and upper bounds on the value of $i$, we can perform a binary search to find $i$ precisely. The procedure is the following, initialising $r=j$:
    \begin{enumerate}
        \item Set $c = \lceil \frac{l+r}{2} \rceil$, and let $J = u_{l},\dots,u_{c}$ be the sub-list of $L$ containing the left half of the vertices indexed between $l$ and $r$. If $|J| = 1$, classically check whether the vertex it contains is good or not. If it is, return $l$, otherwise return $l+1$.
        \item Run \textbf{VertexFind}($J,\zeta'$) to attempt to find a vertex in $J$, which requires at most \\ $O\left(\sqrt{\delta_{\max}(c-l)} \log\left(\frac{(c-l)}{\zeta'}\right) \right)$ function calls and other elementary operations. 
        \item If a marked vertex is found at position $l \leq j \leq c$, then set $r\gets j$. Otherwise, set $l \gets c$. Repeat from step $1$ above.
    \end{enumerate}
    This procedure (which is just binary search on $L$ with a quantum subroutine) will return the index of the left-most good vertex with probability $\geq (1-\zeta')^{\lceil \log(a) \rceil}$ where $a$ is the size of the segment of $L$ identified by the preceding routine. 
    
    Once again, we are running \textbf{VertexFind} (with failure probability now $\zeta'$) on lists that halve in size each time, starting with one that is of size $a/2$. Hence, the total run-time of this part of the algorithm is at most 
    \begin{eqnarray*}
        \sum_{k=0}^{\lceil \log(a/2) \rceil} O\left( \sqrt{\delta_{\max}2^k} \log\left(\frac{2^k}{\zeta'}\right) \right) 
        &\leq& \sum_{k=0}^{\lceil \log(a/2) \rceil} O\left( \sqrt{\delta_{\max}2^k} \log\left(\frac{a}{\zeta'}\right) \right) \\
        &=& O\left( \sqrt{\delta_{\max}a} \log\left(\frac{a}{\zeta'}\right) \right)\,.
    \end{eqnarray*}
    Finally, we note that the segment containing $i$ is of size at most $i$, and hence we find that the run-time of this part of the algorithm is also at most
    \[
        O\left( \sqrt{\delta_{\max}i} \log\left(\frac{i}{\zeta'}\right) \right)\,.
    \]
    
    \ 
    
    It remains to choose the failure probabilities $\zeta$ and $\zeta'$. We require that both parts of the algorithm succeed with probability $\geq \sqrt{(1-\epsilon)}$ each. In order to achieve this, we require for the first part that $(1-\zeta)^q \geq \sqrt{(1-\epsilon)}$, and for the second part that $(1-\zeta')^{\lceil \log(i) \rceil} = (1-\zeta')^q \geq \sqrt{(1-\epsilon)}$. Both conditions can be satisfied by choosing $\zeta = \zeta' = \Omega(\epsilon/q) = \Omega(\epsilon/\log(|L|))$, yielding our final run-time.
\end{proof}
Note that we made use of the worst-case complexities for \textbf{VertexFind} in the above analysis, in particular using the variant of the algorithm that is not faster even when there are more good vertices available, and hence it is likely that in practice the algorithm will be much faster. 

\ 

Finally, we can use the algorithm \textbf{FindFirst} to construct a quantum version of the Louvain algorithm. Recall that the classical algorithm constructs a randomly ordered list of all vertices, locates the first good vertex in this list, and then moves it. Then it chooses the next good vertex, and so on, repeating this process until no good vertices are found in the remainder of the list. The corresponding quantum algorithm is precisely the same as the classical Louvain algorithm, except that the step in which the classical algorithm looks for the next good vertex in the list of vertices is replaced by a single call to the \textbf{FindFirst} algorithm of Lemma~\ref{lem:find_first}.

To see how long this algorithm takes, suppose the classical algorithm makes $T$ moves, and that the $k$th move necessitated inspecting $t_k$ vertices before finding one that could be moved. Then as we saw in Section~\ref{sec:complexity_of_louvain} the classical algorithm will make at most
    \[
        \sum_{k \in [T]} O(\delta_{\max} t_k)\,,
    \]
    calls to $g_\Delta$, and use 
    \[
        \sum_{k \in [T]} \tilde{O}(\delta_{\max} t_k + d_{\max})\,
    \]
other operations.\footnote{Here we assumed that the algorithm makes use of the additional data structure described in Section~\ref{sec:complexity_of_louvain}, which in particular allows us to obtain the inputs to $g_\Delta$ in constant time, in exchange for a $O(d_{\max})$-time update step after moving a vertex.} We will also use this data structure in the quantum version of the algorithm. Finally, note that this run-time is somewhat pessimistic -- the $\delta_{\max}$ could be replaced with the average number of neighbouring communities amongst all vertices inspected during the $k$th step, which in general should be smaller. The quantum algorithm cannot take advantage of this fact, however, and is really limited by $\delta_{\max}$.

\begin{theorem}
    There exists a quantum algorithm \textbf{QLouvain} that, with probability $\geq 2/3$, behaves identically to the Louvain algorithm and requires at most
    \[
        \sum_{k \in [T]}\tilde{O}(\sqrt{\delta_{\max}t_k} )\,
    \]
    calls to $g_\Delta$ and 
    \[
        \sum_{k \in [T]}\tilde{O}(\sqrt{\delta_{\max}t_k} + d_{\max})\,
    \]
    other elementary operations.
\end{theorem}
\begin{proof}
    We replace the part of the classical algorithm that searches for the next good vertex with a single call to \textbf{FindFirst}. Once we identify a good vertex, we can use the \textbf{QMax} algorithm of Lemma~\ref{lem:quantum_max} to obtain the \emph{best} move available for that vertex, with probability at least $1-\epsilon$, using at most $O(\sqrt{\delta_{\max}}\log(1/\epsilon))$ function calls and $O(\sqrt{\delta_{\max}}\log(\delta_{\max})\log(1/\epsilon))$ other operations. After making the move, we have to update the data structure used by the \textbf{VertexFind} subroutine, which incurs a time-cost of $O(d_{\max})$. 

    To obtain the quantum run-time, we first need to choose settings for the failure probabilities of the \textbf{FindFirst} and maximum-finding subroutines. In particular, we require that all calls to both routines succeed with probability $\geq \sqrt{2/3}$ each, and so we can choose the failure probability of both to be $\epsilon = O(1/T) = O(1/\poly(n))$. Hence, the quantum version of Louvain will require at most 
    \[
        \sum_{k \in [T]}O(\sqrt{\delta_{\max}t_k}\log(n))\,
    \]
    calls to $g_\Delta$ and 
    \[
        \sum_{k \in [T]}\tilde{O}(\sqrt{\delta_{\max}t_k} + d_{\max})\,
    \]
    other elementary operations, yielding a square-root worst-case improvement over the classical algorithm. Since the quantum algorithm mimics (with constant probability) the behaviour of the classical one, it will produce exactly the same output. 
\end{proof}

\subsubsection{Simple quantum louvain}
\label{sec:grover_louvain}

Rather than finding the \emph{first} marked item in a list, which requires repeated Grover searches using a bisection method, for a quantum computer it is more natural to simply find \emph{any} marked item, which requires only a single application of Grover search. Motivated by this observation, in this section we describe a slightly different version of the Louvain algorithm that has a much simpler quantum analogue.

Concretely, the simpler algorithm (i) searches over the list of all vertices until it finds a good vertex by sampling vertices at (uniformly) random \emph{with replacement}; (ii) once a good vertex is found, we move it and go back to (i). This process is repeated until no good vertices can be found. The quantum version of this algorithm, which we call \textbf{SimpleQLouvain}, is as follows:
\begin{enumerate}
    \item Call \textbf{VertexFind}($L,\delta$) with $L$ the list of all vertices in the graph, and $\delta$ to be determined. If there are any vertices in the graph for which a good move is available, this will return one at random. Otherwise, it will signal that none exist and we can end this phase of the algorithm.
    \item Assuming that \textbf{VertexFind}($L,\delta$) returned a vertex, we use the \textbf{QMax} algorithm of Lemma~\ref{lem:quantum_max} to obtain the \emph{best} move available to that vertex. 
    \item We move the vertex, update the data structure, and then repeat from Step~1. 
\end{enumerate}
We note that this algorithm is subtly different to the original Louvain algorithm: here we search \emph{with replacement} over the vertices of the graph, whereas the original  algorithm searchers \emph{without replacement} by fixing a randomly ordered list of vertices and sequentially searching through it. We verify numerically in Section~\ref{sec:numerics} that this difference does not qualitatively change the behaviour of the algorithm in any significant way.

\

For a $T$-step run, in order for \textbf{SimpleQLouvain} to succeed (i.e.~find a good vertex whenever one exists) with probability $\geq 2/3$, we require that $(1-\delta)^T \geq 2/3$, which can be achieved by setting $\delta = 1/O(T) = 1/O(\poly(n))$. In that case, if there are an $f_k$ fraction of good vertices in the graph after having already made $k-1$ moves, the expected number of function calls (and other operations) of \textbf{VertexFind} will be at most $\tilde{O}(\sqrt{\delta_{\max}/f_k})$ as per Lemma~\ref{lem:vertex_find}. Finally, we can choose the failure probability $\delta'$ of the quantum maximum-finding routine to be $\delta'=\delta$. Hence, the overall algorithm will require an expected number of at most
\[
    \sum_{k \in [T]} \tilde{O} \left( \sqrt{\frac{\delta_{\max}}{f_k}} \right)\,
\]
calls to $g_{\Delta}$ and 
\[
    \sum_{k \in [T]} \tilde{O} \left( \sqrt{\frac{\delta_{\max}}{f_k}} + d_{\max} \right)\,
\]
other operations. 

In contrast, the expected number of function calls made by the equivalent classical algorithm (the one that also searches with replacement) is
\begin{equation}
	\sum_{k \in [T]} O\left( \frac{\delta_{\max}}{f_k} \right),
	\label{eq:classical_complexity_with_replacement}
\end{equation}
and so the quantum algorithm is asymptotically more efficient for any step of that algorithm (both in terms of calls to $g_\Delta$ and in terms of other operations). In general, the final stages of the algorithm will have $1/f_k = O(n)$, and in these steps the quantum speedup is quite large. 

Finally, we note that if we do not make use of the data structure for the quantum algorithm, then the number of function calls remains the same, but the number of other operations becomes $\sum_{k \in [T]} \tilde{O}\left(\frac{\sqrt{\delta_{\max}}+d_{\max}}{\sqrt{f_k}} + d_{\max} \right)$.

\

\noindent In practice, the $\delta_{\max}$ appearing in the classical run-time in Eq.~\eqref{eq:classical_complexity_with_replacement} is overly pessimistic: it will in fact be closer to the average number\footnote{But not exactly: it is actually the average over subsets of vertices containing precisely one good vertex. If the good vertices have many adjacent communities then this average will be biased towards this. We discuss this in more detail in Section~\ref{sec:QVTAA}.} of adjacent communities, $\delta_{\avg}$, since the algorithm visits vertices one by one, computing $\bar{\Delta}_u$ for each in time $O(\delta_u)$. On the other hand the quantum algorithm really is limited by the maximum number of adjacent communities due to our use of Grover search, and hence in practice could find itself being slower. We observe in Section~\ref{sec:numerics} that in fact this is indeed the case, and so the `worst-case' quantum speedup that we find via an asymptotic analysis often doesn't materialise in practice. In Section~\ref{sec:QVTAA} we describe a more sophisticated quantum algorithm that makes use of the technique of variable time amplitude amplification to remove the dependency on $\delta_{\max}$, in an effort to overcome this limitation.

\subsubsection{Trading a square-Root for a log factor for sparse Graphs}\label{sec:sparse_graph_versions}
In the quantum algorithms described above, the run-time contains a $\log$ factor that could in practice be quite large. For example, the run-time of \textbf{VertexFind}($L$,$\zeta$), which finds a good vertex in a list $L$ (or confirms that there aren't any) with probability at least $1-\zeta$ is $O\left( \sqrt{\frac{\delta _{\max}}{f}} \log\left( \frac{|L|}{\zeta  } \right) \right)$, where $f$ is the fraction of vertices in $L$ that are good. The $\log(|L|)$ overhead arises because the quantum algorithm performs a Grover search over the vertices in $L$, using another Grover search as a subroutine. In order for the outer search to succeed with probability at least $1-\zeta$, the inner Grover search has to succeed with a much larger probability, namely $\approx 1-\frac{\zeta}{\sqrt{|L|}}$ (since the outer search will make in the worst-case $O(\sqrt{|L|})$ calls to the inner search routine). In practice this might be a substantial overhead, especially if $\log(|L|)$ is large relative to $\delta_{\max}$. In our algorithms, the list $L$ will often be of size $\Theta(n)$, and hence for very sparse graphs, for example when $\delta _{\max} \leq d_{\max} = O(\log n)$, this overhead will be large enough to negate the square-root speedup that we obtain in terms of $\delta_{\max}$. 

Hence, for such sparse graphs it will often make sense to replace the inner Grover search with a purely classical routine that succeeds with certainty. We will use the suffix `\textbf{SG}' to signify that the inner loop over the neighbouring communities is classical. In this case we can construct an alternative version of \textbf{VertexFind}, in which the time taken to find a good vertex in $L$ with probability $\geq 1-\zeta$ is now $O\left( \frac{\delta_{\max}}{\sqrt{f}} \log\left( \frac{1}{\zeta} \right) \right)$. Using this variant of \textbf{VertexFind} we can then construct different versions of the above quantum algorithms that might perform better on sparse graphs. It is straightforward to check that the alternative run-times of these new algorithms designed for sparse graphs will be the following.
\begin{itemize}
    \item \textbf{VertexFindSG}($L$,$\zeta$): 
        \begin{itemize}
            \item Expected number of function calls at most 
            \[
                O\left(\frac{\delta_{\max}}{\sqrt{f}} \log\left(\frac{1}{\zeta}\right) \right)\,.
            \]
            \item Worst-case number of function calls 
            \[
                O\left({\delta_{\max}}{\sqrt{|L|}} \log\left(\frac{1}{\zeta}\right) \right)\,.
            \]
        \end{itemize}
    \item \textbf{FindFirstSG}($L$,$\epsilon$): 
        \begin{itemize}
            \item Worst-case number of function calls 
            \[
                O\left({\delta_{\max}}{\sqrt{|L|}} \log\left(\frac{\log(|L|)}{\epsilon}\right) \right)\,.
            \]
        \end{itemize}
    \item \textbf{QLouvainSG}: 
        \begin{itemize}
            \item Worst-case number of function calls, if classical Louvain makes $T$ moves, with $t_k$ vertices inspected during move $k$:
            \[
                \sum_{k \in [T]} O\left({\delta_{\max}}{\sqrt{t_k}} \log\left(T\log(n)\right) \right)\,.
            \]
        \end{itemize}
    \item \textbf{SimpleQLouvainSG}: 
        \begin{itemize}
            \item Expected number of function calls, if the algorithm makes $T$ moves, with $f_k$ the fraction of good vertices available during move $k$:
            \[
                \sum_{k \in [T]} O\left(\frac{\delta_{\max}}{\sqrt{f_k}} \log\left(T\right) \right)\,.
            \]
        \end{itemize}
\end{itemize}
Hence, the $\log(n)$ factor present in \textbf{QLouvain} becomes a $\log\log(n)$ factor, although the overhead of $\log(T)$ is still present in the new versions of both that algorithm and of \textbf{SimpleQLouvain}.

\subsubsection{Quantum Louvain via variable time amplitude amplification}
\label{sec:QVTAA}

An unsatisfactory element of the algorithms from the previous sections is that the subroutine that computes $\bar{\Delta}_u$ takes a different amount of time for each $u$, but the outer Grover search of \textbf{VertexFind} is limited by its slowest branch and hence its run-time depends on $\delta_{\max}$, in contrast to the classical algorithm whose run-time depends on a value closer to $\delta_{\avg}$. For many families of graphs (e.g. power-law graphs), this discrepancy could be significant -- i.e. it might not be unlikely that $\sqrt{\delta_{\max}} > \delta_{\avg}$. In this section we describe a more sophisticated quantum algorithm, \textbf{VertexFindVTAA}($L$,$\zeta$), that searches for good vertices from a list $L$, and whose run-time is sensitive to the fact that most vertices will not have a number of neighbouring communities close to the maximum. Similarly to the previous sections, we can then use this quantum algorithm as a subroutine to construct quantum algorithms for community detection. 

Our main technical tool is the \emph{variable time amplitude amplification} algorithm of Ambainis~\cite{ambainis12}. Using this as a subroutine, we show
\begin{restatable}{theorem}{VTAA}
\label{theo:VTAA}
Given a list $L$ of vertices such that a fraction $f>0$ of them are good, and the unitaries $\mA_c$ and $\mA_s$ defined in Eqs.~\eqref{eq:A_s} and~\eqref{eq:A_c}, we can use variable time amplitude amplification to construct a quantum algorithm \textbf{VertexFindVTAA}($L,\zeta$) that makes an expected
\[
    O\left(\left(\delta_{\max} \log(\delta_{\max}) + \frac{t_{\text{avg}}^\text{q}}{\sqrt{f}} \log^{1.5}\delta_{\max}\right) \log(1/\zeta)\right)
\]
calls to $g_\Delta$, where
\[
    t_{\avg}^{\text{q}} = \sqrt{\sum_{i=1}^{\delta_{\max}} p_i i^2},
\]
and that returns the identity of a good vertex and the best move available to it with probability $\geq 1-\zeta$. If there is no good vertex, the algorithm will signal this and requires at most 
\[
    O\left(\left(\delta_{\max} \log(\delta_{\max}) + t_{\text{avg}}^\text{q}\sqrt{|L|} \log^{1.5}\delta_{\max}\right) \log(1/\zeta)\right)
\]
queries to do so.
\end{restatable}
\noindent We defer to Appendix~\ref{app:VTAA} for details of the algorithm and the proof of the theorem. 

\

\noindent Note that by using this version of the \textbf{VertexFind} algorithm, we lose the square-root improvement of the dependence on $\delta_{\max}$ that we obtained with the simpler Grover-based quantum algorithm. In exchange for the square-root speed up, the dependence on $\delta_{\max}$ is improved to a dependence on something closer to $\delta_{\avg}$. The reason for losing this speed up is because we exchanged a quantum search over the neighbours of each vertex with a classical, sequential one. One might wonder whether we could retain the square-root speedup for this part of the algorithm by replacing the classical algorithm $\cA = \cA_c \cdots \cA_c \cA_s$ with a quantum subroutine, say, $\cQ = \cQ_q \cdots \cQ_q$. However, this does not seem possible since existing quantum algorithms for search with an unknown number of marked items, such as the Grover search of Lemma~\ref{lem:grover}, cannot be separated into fixed `steps' $\cQ_q$ that satisfy the conditions from~\cite{ambainis2010quantum} (and described in the proof of the theorem above) without assuming, for example, that every vertex has the same number of good moves -- that is, to use VTAA with a subroutine $\cQ$, the algorithm $\cQ$ cannot be \emph{adaptive}, and must instead act identically on every branch of the superposition to which it is applied (in our case, on each vertex $u\in V$). It is an interesting open question whether the techniques from~\cite{ambainis2010quantum} can be extended to such an adaptive setting, particularly in the case where the subroutine used in VTAA is a quantum search over an unknown number of marked items.

Finally, we remark that, due to its very complex nature, it was extremely difficult to precisely pin down the exact number (including constants) of function calls required by the VTAA-based quantum algorithm.\footnote{In fact, we suspect that it is not possible to do so (without substantial work) from the description given in~\cite{ambainis2010quantum} alone.} For this reason we did not numerically simulate the algorithm in order to compare it to our other quantum (and classical) algorithms for community detection.

\subsubsection{Quantum Louvain algorithms utilizing different search-spaces}
\label{sec:EdgeLouvain}

In this section we consider a vastly simplified algorithm, called \textbf{EdgeQLouvain} which only utilizes a single Grover search over a large search space, rather than searching over multiple search spaces (vertices and their neighbours).\footnote{We also considered a similar algorithm that instead searches over (vertex, neighbouring-community) pairs. This algorithm gave very similar results to \textbf{EQL}, and therefore we restrict our attention only on \textbf{EQL}.} The algorithm differs somewhat in spirit to the original Louvain algorithm described in Section~\ref{sec:classical}, however as we show in Section~\ref{sec:numerics}, tends to yield similar results in terms of the modularity it obtains. 

\textbf{EdgeQLouvain}, or \textbf{EQL} for short, searches over directed edges $(u,v)$ of the graph for one that gives an increase in modularity if $u$ is moved to community $\ell(v)$. Upon finding one, it moves $u$ greedily to a neighboring community (but not necessarily to $\ell(v)$ itself). This approach has the two advantages that, in the quantum case, the Grover search does not need to make use of another nested Grover search over neighboring communities, and likewise there will now be only an additive dependency on either one of $\delta_{\max}$ or $\delta_{\avg}$, nullifying the two issues that we have encountered with our algorithms thus far. 

Given the edge set $E$ of the input graph for the Louvain algorithm, let $E_d = \{(u,v): \{u, v\} \in E \}$ be the set of directed edges obtained by replacing every undirected edge $\{u,v\} \in E$ by both $(u,v)$ and $(v,u)$, making $|E_d| = 2|E|$. As usual, we assume that we have access to the data structure described in Section~\ref{sec:complexity_of_louvain}. The three phases of the algorithm operate as follows:

\ 

\noindent \textit{Initialization} -- Exactly the same initialization procedure as in OL (original Louvain), see Section~\ref{sec:OL_init} for details. 

\ 

\noindent \textit{First phase} -- Use \textbf{QSearch} to search over all edges $(u,v)$ in search of one that yields a good move. As an oracle we provide the unitary that, for a pair $(u,v)$, computes whether $\Delta_{u}^{\ell(v)}$ is positive or not, which can be done using $O(1)$ function calls to $g_\Delta$ and $O(\log \delta_u)$ other operations (to obtain the inputs to $g_\Delta$ we need to perform a binary search over the community adjacency list of $u$ to find the entry corresponding to the community of $v$, whilst the other inputs can be obtained in constant time).

Instead of moving $u$ to $\ell(v)$, we find the \emph{best} neighbouring community of $u$, $\bar{\alpha} = \argmax_{\alpha \in \zeta_u} \Delta_u^{\alpha}$, by using quantum maximum finding over all moves to neighbouring communities of $u$. We then update the data structure as we do for OL in time $\sum_{v \in N_u} O(\log \delta_v)$ (see Section~\ref{sec:classical}). 

\ 

\noindent  \textit{Second phase} --  Again this is identical to the second phase of OL; see Section~\ref{sec:OL_second_phase}.

\ 

By using \textbf{QSearch}, we can find a pair $(u,v)$ yielding a good move with an expected 
\[
    O \left(\frac{1}{\sqrt{h_k}} \log(1/\delta)\right)
\]
number of calls to $g_\Delta$, and 
\[
    \tilde{O} \left(\frac{1}{\sqrt{h_k}} \log(1/\delta)\right)
\]
other operations, where $h_k$ is the fraction of edges that yield a good move during the $k$th step. If no such pair exists then the algorithm will signal this after making at most $O(\sqrt{|E|}\log(1/\delta))$ queries to $g_\Delta$. 

Once we have found such a vertex, with probability $\geq 1-\epsilon$ we can find the best move available using the \textbf{QMax} algorithm of Lemma~\ref{lem:quantum_max}, which will require at most $O(\sqrt{\delta_{\max}} \log(1/\epsilon))$ calls to $g_\Delta$ and $O(\sqrt{\delta_{\max}} \log(\delta_{\max})\log(1/\epsilon))$ other operations. If the algorithm makes $T$ moves in total, we will need to choose $\epsilon$ and $\delta$ such that all calls to either \textbf{QSearch} or quantum maximum-finding will succeed with sufficiently high likelihood that the probability that any one of them fails is at most $2/3$, which can be satisfied by choosing $\epsilon=\delta=1/O(T)=1/O(\poly n)$. Hence, the algorithm will make an expected number of function calls at most
\[
    \sum_{k \in [T]} \tilde{O} \left(\frac{1}{\sqrt{h_k}} + \sqrt{\delta_{\max}} \right)
\]
and 
\[
    \sum_{k \in [T]} \tilde{O} \left(\frac{1}{\sqrt{h_k}} + \sqrt{\delta_{\max}} + d_{\max} \right)
\]
other operations.

The complexity of this algorithm appears to be favourable compared to the algorithms from the previous sections, and has the additional advantage of being very simple and therefore incurring smaller logarithmic overheads. As we show in Section~\ref{sec:numerics}, it also behaves similarly to the original Louvain algorithm in practice, whilst being the fastest amongst all quantum algorithms that we evaluated.

\section{Estimating the run-times of quantum algorithms for community detection}
\label{part:three}

In this section and the next we use the tools and methodology of~\cite{ourotherpaper} to empirically estimate the run-times (more precisely the number of queries to the (gradient of) the modularity function) of our quantum algorithms for a variety of inputs, and use these estimates to compare their performances, to each other and to their classical counterparts. This allows us to estimate how much of the per-step speedup suggested by the asymptotic analyses in Section~\ref{sec:quantum} manifests in the final behaviour of the algorithms, for a range of inputs. We find for all algorithms that some speedup does make it out, though to varying degrees. Moreover we observe that the algorithms that promise the greatest speedups through an asymptotic analysis are not necessarily the ones that achieve the best speedups `in practice'. In our view, this demonstrates the usefulness of this sort of analysis over a purely asymptotic, worst-case one for designing efficient quantum algorithms to use for practical tasks.

In the sections that follow we describe explicitly our approach to simulating our quantum algorithms and estimating their expected run-times. We will focus on \textbf{QLouvain}, \textbf{SimpleQLouvain} and \textbf{EdgeQLouvain} (including both their original and sparse-graph versions), all of which fit into the framework of Algorithm~\ref{alg:general} introduced in Section~\ref{sec:introduction}. We deliberately chose to forgo simulating the algorithms that make use of variable time amplitude amplification (VTAA) as a subroutine, not only because the nature of VTAA makes it difficult to do so, but also because the expected speedup will only be a constant given that the Louvain algorithm is predominantly applied to sparse graphs.

\subsection{Complexity bounds}\label{sec:complexity_bounds}
The first step is to obtain tight bounds (including all constants etc.)~on the complexities of the quantum sub-routines that we make use of. We begin by recalling the complexity bounds obtained in~\cite{ourotherpaper} for the two quantum sub-routines that we use here: Grover search with an unknown number of marked items (\textbf{QSearch}, Lemma~\ref{lem:grover}), and quantum maximum-finding (\textbf{QMax}, Lemma~\ref{lem:quantum_max}).

\paragraph{Expected complexity of Grover search} As we mentioned in Section~\ref{sec:prelims}, when considering the full run-time, including constants, of \textbf{QSearch}, there is an extra hyper-parameter $N_{\text{samples}}$ used to determine the number of classical samples that are drawn before Grover search is used. Then the worst-case expected complexity of \textbf{QSearch} is as follows

\begin{lemma}[Worst-case expected complexity of \textbf{QSearch}, {[Lemma 4,~\cite{ourotherpaper}]}] 
\label{lem:QSearch}
Let $L$ be a list, $g: L \rightarrow \{0,1\}$ a Boolean function, $N_{\text{samples}}$ a non-negative integer and $\epsilon > 0$, and write $t = |g^{-1}(1)|$ for the (unknown) number of marked items of $L$. Then, $\textbf{QSearch}(L, N_{\text{samples}},\epsilon)$ finds and returns an item $x \in L$ such that $g(x) = 1$ with probability at least $1-\epsilon$ if one exists using an expected number of queries to $g$ that is given by
\begin{equation}\label{eq:e_qsearch}
    E_{\textbf{QSearch}}(|L|,t,N_{\text{samples}},\epsilon) = \frac{|L|}{t}\left(1 - \left(1-\frac{t}{|L|}\right)^{N_{\text{samples}}}\right) + \left(1-\frac{t}{|L|}\right)^{N_{\text{samples}}} c_q Q_{\text{Grover}}(|L|,t) \, ,
\end{equation}
where
\begin{equation}
    Q_{\text{Grover}}(|L|,t)
    \leq F(|L|,t) \left(1 + \frac{1}{1 - \frac{F(|L|,t)}{\alpha \sqrt{|L|}}} \right) \, ,
\label{eq:QGrover_ubound_methodology}
\end{equation}
with
\begin{equation}\label{eq:flt}
    F(|L|,t) =
    \begin{cases}
        \frac{9}{4}\frac{|L|}{\sqrt{(|L| -t )t}} + \left\lceil\log_{\frac{6}{5}}\left(\frac{|L|}{2\sqrt{(|L| -t )t}} \right) \right\rceil - 3 \leq \frac{\alpha \sqrt{L|}}{3 \sqrt{t}}   &\text{for} \quad 1 \leq t < \frac{|L|}{4} \\
        2.0344 &\text{for} \quad \frac{|L|}{4} \leq t \leq |L|. 
    \end{cases}
\end{equation}
If no marked item exists, then the expected number of queries to $g$ equals the number of queries needed in the worst case (denoted by $W_{\textbf{QSearch}}(|L|, N_{\text{samples}},\epsilon)$), which is given by
\begin{equation}
    E_{\textbf{QSearch}}(|L|,0,N_{\text{samples}},\epsilon) = W_{\textbf{QSearch}}(|L|,N_{\text{samples}},\epsilon) \leq N_{\text{samples}} + \alpha c_q \ceil{\log_3(1/\epsilon)}) \sqrt{|L|} \, . \label{eq:worst_qsearch} \\
\end{equation} 
In the formulas above, $c_q$ is the number of queries to $g$ required to implement the oracle $\mO_g \ket{x}\ket{0} = \ket{x}\ket{g(x)}$, and $\alpha = 9.2$.
\end{lemma}

\paragraph{Worst-case complexity of Grover search} Using a modified version of \textbf{QSearch} (which we call \ZQ, described in~\cite{ourotherpaper} and based on the algorithm in~\cite{zalka1999grover}), we can obtain an algorithm with better complexity in the case of no marked items.
\begin{lemma}[worst-case complexity of \ZQ, {[Lemma 5,~\cite{ourotherpaper}]}]
Let $L$ be a list of items, $g:L \rightarrow \{0,1\}$ a Boolean function and $\epsilon > 0$, and write $c_q$ for the number of queries to $g$ required to implement the oracle $\mO_g \ket{x}\ket{0} = \ket{x}\ket{g(x)}$. Then, with probability of failure at most $\epsilon$, \ZQ requires at most
\begin{equation}\label{eq:qsearch_max}
    W_{\textbf{QSearch}_{\text{Zalka}}}(|L|,\epsilon) := c_q \left(5\ceil*{\frac{\ln (1/\epsilon)}{2\ln(4/3)}} + \pi \sqrt{|L|} \sqrt{\ceil*{\frac{\ln (1/\epsilon)}{2\ln(4/3)}}} \right)\,
\end{equation}
queries to $g$ to find a marked item of $L$, or otherwise to report that there is none. 
\end{lemma}

\paragraph{Quantum maximum-finding} For maximum finding, we have the following result from~\cite{ourotherpaper}.
\begin{lemma}[Expected complexity of \textbf{QMax}, {[Corollary 1,~\cite{ourotherpaper}]}]
Let $L$ be a list of items of length $|L|$ and $R : L \rightarrow \mathbb{R}$ a function that assigns a value to each item. Let $f_i$ be the marking function
\[
    f_i(j) = 
    \begin{cases}
        1 &\text{if} \quad  R(j) > R(i) \\
        0 &\text{otherwise} \, .
    \end{cases}
\]
Then the expected number of queries to $f_i$ (for any $i$) required for \textbf{QMax} to find the maximum of $L$ with success probability at least $1-\epsilon$ is $\ceil{\log_3(1/\epsilon)}3 E_{\qmi}(|L|)$, where 
\begin{equation}
    E_{\qmi}(|L|) \leq c_q \sum_{t=1}^{|L|-1} \frac{F(|L|,t)}{t+1} \, ,
    \label{eq:EQMax-inf}
\end{equation}
with $F(|L|,t)$ given by \eqref{eq:flt}, and where $c_q$ is the number of queries to $f_i$ required to implement oracle access to $f_i$ (which we assume to be the same for all $i$).
\end{lemma}

\ 

\noindent Using these bounds, we proceed to obtain bounds on the expected complexities of the main quantum subroutines used by our quantum community detection algorithms -- \textbf{VertexFind} and \textbf{FindFirst}. Afterwards, in Section~\ref{sec:simulation_details} we describe precisely how we simulate our quantum community-detection algorithms, including what we implement classically and what information we gather along the way, as well as what accuracy and hyperparameter settings we use.

\subsubsection{Expected complexity of \textbf{VertexFind}}\label{sec:expected_vertexfind}

To start with, we bound the expected complexity of the \textbf{VertexFind} algorithm when it is run on a list $L$ containing $|L|$ vertices, $t$ of which are good vertices, and when it is required to fail with probability at most $\zeta$. Following the analysis in the proof of Lemma~\ref{lem:vertex_find}, we run \textbf{QSearch}($L$,$N_{\text{samples}}$,$\epsilon$), giving it access to a subroutine \ZQ($J$,$\epsilon'$), where $J$ will be a list of length at most $\delta_{\max}$, and $\epsilon$ and $\epsilon'$ are parameters that are determined by $\zeta$ and (in the case of $\epsilon'$) the worst-case complexity of the outer \textbf{QSearch}($L$,$N_{\text{samples}}$,$\epsilon$) routine. Note that, in order to implement the oracle required for Grover search, the subroutine \ZQ($J$,$\epsilon'$) and its inverse will each be run once per query. 

On a list $L$ of size $|L|$ and with failure probability at most $\epsilon$, the \emph{outer} \textbf{QSearch} routine that we use requires in the worst case (i.e. when $t=0$) at most $W_{\textbf{QSearch}}(|L|,N_{\text{samples}},\epsilon)$ queries to its oracle/subroutine (see Eq.~\eqref{eq:worst_qsearch}). When there are $t > 0$ marked items, it requires an expected $E_{\textbf{QSearch}}(|L|,t,N_{\text{samples}},\epsilon)$ queries (see Eq.~\eqref{eq:e_qsearch}). The \emph{inner} \ZQ routine is slightly different, and it requires in the worst case at most $W_{\textbf{QSearch}_{\text{Zalka}}}(|J|,\epsilon')$ queries on a list $J$ of size $|J|$ and with failure probability at most $\epsilon'$ (see Eq. \eqref{eq:qsearch_max}).

Finally, we need to set the failure probabilities appropriately to align them with the overall failure probability $\zeta$ for \textbf{VertexFind}. From Lemma~\ref{lem:vertex_find} we find that to achieve a success probability $\geq 1-\zeta$, we can set $\epsilon = \zeta/2$ and $\epsilon' = \frac{\zeta}{2W_{\textbf{QSearch}}(|L|,N_{\text{samples}},\epsilon)}$. Putting everything together, the expected complexity of the entire \textbf{VertexFind} algorithm will be at most
\begin{eqnarray}
\hspace{-2.3cm}
    &E_{\textbf{VertexFind}}(|L|,t,N_{\text{samples}},\zeta) \rule{10.5cm}{0cm} \nonumber \\
    &= E_{\textbf{QSearch}}(|L|,t,N_{\text{samples}},\zeta/2)  \cdot 2W_{\textbf{QSearch}_{\text{Zalka}}}\left(\delta_{\max},\frac{\zeta}{2W_{\textbf{QSearch}}(|L|,N_{\text{samples}},\zeta/2)}\right) \nonumber \label{eq:e_vertexfind} \\
    &\leq E_{\textbf{QSearch}}(|L|,t,N_{\text{samples}},\zeta/2)
    \cdot 2c_q \Bigg[5 \left\lceil 
    {\ln\left( 
        \frac{2W_{\textbf{QSearch}}(|L|,N_{\text{samples}},\zeta/2)}{\zeta} 
    \right)}/{(2 \ln (4/3))} \right\rceil   \nonumber\\
    &\phantom{=}
    \phantom{E_{\textbf{QSearch}}(|L|,t,\zeta/2)\cdot\Bigg(c } 
    + \pi \sqrt{\delta_{\max}} \sqrt{ \ceil*{ {\ln\left(\frac{ 2W_{\textbf{QSearch}}(|L|,N_{\text{samples}},\zeta/2)}{\zeta} \right)}\big/{(2 \ln(4/3))} } }\Bigg]\,.
\end{eqnarray}
where the expression for $E_{\textbf{QSearch}}$ can be found in Eq. \eqref{eq:e_qsearch} and the expression for $W_{\textbf{QSearch}}$ in Eq. \eqref{eq:worst_qsearch}.

\subsubsection{Expected complexity of \textbf{VertexFindSG}}\label{sec:expected_vertexfind}
Next, we bound the expected complexity of the sparse-graph version of \textbf{VertexFind} (described in Section~\ref{sec:sparse_graph_versions}) when it is run on a list $L$ containing $|L|$ vertices, $t$ of which are good, and when it is required to fail with probability at most $\zeta$. This time the inner \ZQ routine is eliminated, and hence the complexity of this algorithm depends only on a single application of \textbf{QSearch}, given access to a classical sub-routine that makes $\delta_{\max}$ queries per call. Similarly to the above case, this classical algorithm must be run twice in order to implement the oracle required for the \textbf{QSearch} routine. The expected complexity of the entire \textbf{VertexFind} algorithm will now be at most
\begin{eqnarray}
\hspace{-1cm}
    E_{\textbf{VertexFindSG}}(|L|,t,N_{\text{samples}},\zeta) &=& E_{\textbf{QSearch}}(|L|,t,N_{\text{samples}},\zeta) \cdot 2 \delta_{\max}\,. \label{eq:e_vertexfindSG}
\end{eqnarray}

\subsubsection{Expected complexity of \textbf{FindFirst} and \textbf{FindFirstSG}}\label{sec:findfirst_worst_case}
The \textbf{FindFirst} routine makes a number of repeated calls to \textbf{VertexFind}, whose expected complexity on a list $L$ with $t$ marked items, and failure probability at most $\zeta$, is given by $E_{\textbf{VertexFind}}(|L|,t,N_{\text{samples}},\zeta)$ from Eq. \eqref{eq:e_vertexfind}. In order to choose the correct setting of $\zeta$ to ensure that \textbf{FindFirst} fails with probability at most $\mu$, we need to know the maximum number of times that \textbf{VertexFind} might be called by \textbf{FindFirst}. The worst case is when there is a single marked item lying at the very end of the list $L$. In this case, \textbf{FindFirst} will make $\lceil\log_2 |L| \rceil$ calls to \textbf{VertexFind}, on sets of increasing size, followed by a binary search that will require another $\lceil\log_2 \frac{|L|}{2} \rceil$ calls. In total \textbf{FindFirst} will call \textbf{VertexFind} at most $2\lceil\log_2 |L| \rceil - 1$ times, and hence to ensure that it fails with probability $\leq \mu$ we must ensure that every call to \textbf{VertexFind} is made with failure probability at most $\zeta \leq \frac{\mu}{2\lceil\log_2 |L| \rceil}$. Precisely the same analysis holds for \textbf{FindFirstSG}, which will make calls instead to \textbf{VertexFindSG}, again with failure probability at most $\zeta \leq \frac{\mu}{2\lceil\log_2 |L| \rceil}$.

\subsection{Simulation details}
\label{sec:simulation_details}
In order to simulate any of our quantum community detection algorithms, we run a corresponding classical version of the algorithm of interest, and collect the information necessary to estimate how long the quantum algorithm would have taken in expectation. The classical algorithms we use are all based on an implementation of Louvain in Python by Aynaud~\cite{python_louvain}.

During Phase 1 of (any version of) the Louvain algorithm (the part that is most time consuming and also benefits from a quantum speedup), we obtain our estimates for the complexities of our quantum algorithms by applying the upper bounds on the expected complexities of the \textbf{QSearch} and \textbf{VertexFind} algorithms given in Section~\ref{sec:complexity_bounds} above. Our upper bounds depend on four parameters: the size $|L|$ of the list $L$ to which the subroutines are applied, the desired failure probability $\zeta$ of the algorithm, the number $t$ of marked vertices in the list, and the choice of hyper-parameter $N_{\text{samples}}$. In the sections below, we discuss how to obtain values for these parameters. In particular:
\begin{itemize}
    \item The sizes $|L|$ of the lists will be available during the course of the simulation.
    \item The number $t$ of marked vertices, however, is not immediately accessible -- we discuss below in Section~\ref{sec:num_of_good_vertices} how to obtain this value during the execution of the classical Louvain algorithm.
    \item The failure probabilities $\zeta$ for each sub-routine are determined by the overall acceptable failure probability for the entire algorithm, which requires knowing how many times each subroutine will be called. Of course, this information is not known ahead of time, and we discuss this in Section~\ref{sec:num_moves}.
    \item $N_{\text{samples}}$ is a hyper-parameter that changes the efficiency of the algorithms. Its optimal setting is discussed in Section~\ref{sec:hyper_parameters}
\end{itemize}

Finally, in order to be able to estimate the number of queries used by the different quantum algorithms introduced in Section~\ref{sec:quantum}, the corresponding implementations of the classical Louvain algorithm differ in places to that of Aynaud's. The differences in each case concern how precisely we find a vertex to move to a new community, and we discuss our implementations in Section~\ref{sec:finding_marked_vertex}.

\subsubsection{Computing the number of marked items}
\label{sec:num_of_good_vertices}
Whenever we want to know how many queries a call to \textbf{QSearch} uses, we need to know the number of marked items. Below we describe our methods for obtaining this value while we simulate the quantum algorithm.

For the Louvain algorithm, the number of marked items can be kept track of exactly with the help of an additional data structure in the form of a set. We call this set $V_{\text{marked}} = \{u\in V | \exists \alpha \text{ such that } \Delta^{\alpha}_u > 0\}$, the set of marked vertices, i.e. those that have a neighbouring community they can move to in order to increase the modularity. With this set we can easily compute $t = |V_{\text{marked}}|$. The simplest way to obtain $V_{\text{\text{marked}}}$ is by performing an exhaustive search over the entire list of vertices and explicitly checking for every vertex if it is marked or not (i.e. has a good move). This exhaustive search should be repeated after every move, since a move will change which vertices are marked and which are not. An exhaustive search after every move is tractable for small graphs, but quickly becomes intractable as the number of vertices increases.  

Instead of recreating the list $V_\text{marked}$ from scratch after every move, it can be initialized at the start and then updated after every move. This can be done by searching through the set of vertices that a particular move could possibly affect. More precisely, let $u$ be the vertex that was moved from community $\alpha$ to community $\beta$ in the previous step. The only vertices that can be flipped from marked to non-marked and vise versa are those that belong to either community $\alpha$ or $\beta$ or are neighbours of these communities. More specifically, the vertices contained in $V_{\text{changeable}} = N_{C_\alpha} \cup N_{C_\beta}$, where $N_{C_\alpha} = \{w \in N_v| v \in C_{\alpha}\}$ and likewise for $N_{C_\beta}$. This set of vertices is relatively small compared to the set of all nodes $|V| >> |V_{\text{changeable}}|$, which gives a fast update procedure for $V_\text{marked}$. After every move we construct the set $V_{\text{changeable}}$ as described above. For all vertices in $V_{\text{changeable}}$ we calculate $\Delta_u$. If $\Delta_u >0$, and if $u$ is not in $V_{\text{marked}}$, we add $u$ to $V_{\text{marked}}$. Similarly, if $\Delta_u \leq 0$, and $u$ is in $V_{\text{marked}}$, we remove $u$ from $V_\text{marked}$.

The added benefit of creating and maintaining this list, besides explicitly knowing $t$, is that we can directly sample from it to get a marked vertex. This in fact speeds up the search process of Louvain tremendously\footnote{One could think that this would be a faster method for classical Louvain as search is done in a constant number of steps. It turns out that this is not the case, as updating the data structure still requires a lot of steps.}. At every move we can now use $t=|V_\text{marked}|$ as an input to the bounds for the number of queries that our quantum algorithms would have made. This method is especially fast when the input graphs are sparse.

When searching over edges, as is done in \textbf{EdgeQLouvain}, we use the same procedure, but for marked edges rather than vertices.

It should be noted that, as discussed in~\cite{ourotherpaper}, there are also sampling methods to estimate the number of marked items. We found that in the particular case of the Louvain algorithm, the sampling methods actually were not more efficient that simply keeping track of the aforementioned data structure, and therefore we have not included our results obtained through sampling in this paper. We do find that the number of function calls obtained using either the exact method described above or the sampling method described in~\cite{ourotherpaper} give total query counts that are quantitatively the same.

\subsubsection{Classical simulation of the various Louvain algorithms}\label{sec:finding_marked_vertex}

The program that simulates \textbf{QLouvain} runs the original Louvain with the addition of a subroutine to keep track of the number of queries. This subroutine is called whenever the first marked vertex is found. It estimates the number of function calls that \textbf{QLouvain} would have made to find this vertex. To do so, it simulates the behaviour of \textbf{FindFirst} ,as explained in \ref{sec:finding_marked_vertex}, by searching over lists with varying sizes. For every list $L$ with length $|L|$ smaller than a certain $|L|_{switch}$, a hyper parameter described in Section~\ref{sec:hyper_parameters}, \textbf{FindFirst} uses classical search to find the marked item, else it uses \textbf{VertexFind}. The parameters that are needed to calculate the number queries that \textbf{VertexFind} makes are set as follows: $|L|$ is the size of the respective list, $\zeta$ (the precision) is set as described in Section~\ref{sec:num_moves}, and $t$ is calculated directly from the list by the use of an exhaustive search. Adding these estimates together gives an estimate  of the queries that \textbf{FindFirst} uses per move. The total number of queries is the sum of queries made at every move. During a single run we calculate the number of queries made by both \textbf{VertexFind} and \textbf{VertexFindSG}, using the same parameters, so that we get an estimate for both \textbf{QLouvain} and \textbf{QLouvainSG} in one pass.

\textbf{SimpleQLouvain} uses the subroutine \textbf{VertexFind} on the set of all vertices to find a random marked one. For this we can use our method described in Section~\ref{sec:num_of_good_vertices} to obtain a marked vertex \textit{and} an explicit calculation of $t$. After every move we compute the number of function calls that \textbf{VertexFind} would have made by using the following parameters: $|L| = |V|$ the number of vertices in the graph, $\zeta$ (the precision) is set as described in Section~\ref{sec:num_moves}. In the same simulation we also calculate upper bounds for \textbf{VertexFindSG}, using the same parameters, so that again we get an estimate of the number of function calls for both \textbf{SimpleQLouvain} and \textbf{SimpleQLouvainSG} in a single pass.

\textbf{EdgeQLouvain} uses the subroutine \textbf{Qsearch} on the set of edges to find a marked edge. Similar to \textbf{SimpleQLouvain}, this can be simulated by keeping a list of all marked edges, as described in Section~\ref{sec:num_of_good_vertices}. After every move we estimate the number of queries \textbf{Qsearch} would have made using the parameters: $L = 2 |E|$ ($|L| = 2 |V| \delta_{\text{avg}}$), because the search is over directed edges as described in Section~\ref{sec:EdgeLouvain}, $\zeta$ (the precision) is set as described in \ref{sec:num_moves}, and $t$ is obtained using the methods described in Section~\ref{sec:num_of_good_vertices}.

Finally, once we have identified a good vertex $u$ found by either \textbf{VertexFind}, \textbf{FindFirst} or \textbf{QSearch} (in case of \textbf{EdgeQLouvain}), we have to determine to what neighboring community $u$ should be moved by computing $\argmax_{\alpha \in \zeta_u}(\Delta_u^{\alpha})$. We can find the maximum either classically, or with the quantum maximum-finding subroutine of Lemma~\ref{lem:quantum_max}. Since we keep track of the neighbouring communities in our data-structure, we known $\delta_u$, and therefore we can decide beforehand whether, in expectation, it would be faster to either perform classical maximum-finding, or quantum maximum-finding. For the sparse graphs simulated in the results section, it turns out that classical maximum-finding always uses fewer queries, and quantum maximum-finding is therefore never used.

\subsubsection{Setting the failure probabilities of the subroutines}
\label{sec:num_moves}

One piece of information that we will not be able to obtain without running the entire algorithm first is the total number of moves that the algorithm will make, which determines the maximum number of times each subroutine might be called, which in turn determines the maximum acceptable failure probability for those subroutines. Of course, this is not a problem unique to our use-case, but to any speedup of such a heuristic algorithm. One option is to use the trivial upper bound from Appendix~\ref{app:louvain_total_moves} of $O(\poly(n))$ total moves, however this is almost certainly a huge over-estimate, and in fact in practice it is observed that the Louvain algorithm makes only about $O(n\log n)$ moves before stopping~\cite{lancichinetti2009community}. In Appendix~\ref{sec:number_moves} we present data that further justifies using such an estimated bound on the number of moves.

The only practical way to deal with this issue is to decide beforehand on an upper bound $M$ to the number of moves $T$, perhaps based on empirical observations, and then use this to set the failure probabilities of the various subroutines. If the actual number of moves goes beyond $M$, then we conclude that we can no longer guarantee that the algorithm ran successfully. However, since the algorithm is anyway a heuristic, and the success of the algorithm is determined on more of a qualitative level than a quantitative one, some small number of failures (say, in choosing the `wrong' vertex to move, or not finding the true maximum amongst all possible moves for a single vertex) may be entirely acceptable. In Section~\ref{sec:numerics} we give empirically determined estimates for $M$ based on the size of the input graph, and then use these to derive sensible failure probabilities for the main quantum subroutines used for each algorithm. These then determine the failure probabilities for the other subroutines, as described in Section~\ref{sec:complexity_bounds}.

\subsubsection{Hyper-parameter settings}
\label{sec:hyper_parameters}

\paragraph{Optimal number of $N_{\text{samples}}$} As discussed, the number of classical samples used in the implementation of \textbf{QSearch} is a hyper-parameter that can be tuned to optimize the run-time depending on the size of the list $L$, and what fraction $f$ of the items in $L$ are marked. Classical sampling requires fewer queries than Grover search does when a large fraction of the items is marked, whereas it is more efficient to not use classical sampling at all when a small number of them is marked. In~\cite{ourotherpaper}, it was found that (with $f_0$ being the value of the fraction for which the expected number of queries for Grover search and classical sampling are equal)
\begin{itemize}
    \item For $|L| \leq 260$, classical sampling always requires fewer queries.
    \item For $|L| \geq 260$, the point $1/f_0$ for which Grover search becomes more efficient than classical sampling is plotted as a function of $|L|$ in Fig.~\ref{fig:classical_vs_Grover_1f}.
    \begin{figure}[htb]
        \centering
        \centering
        \includegraphics[scale=0.65]{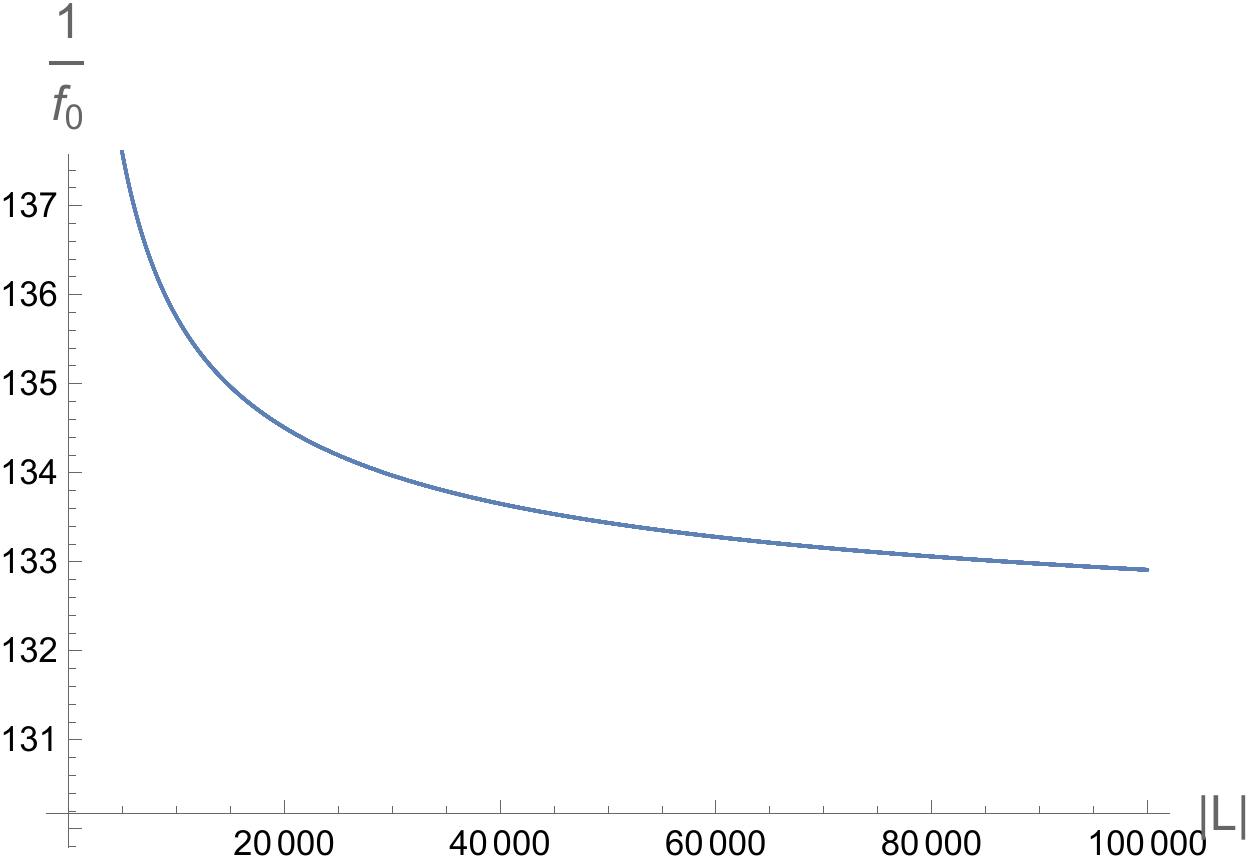}
        \caption{The value of $1/f_0$ as a function of the list length $|L|$ that marks the point beyond which, in expectation, Grover search requires fewer queries than sampling classically does. In the limit $|L| \rightarrow \infty$, there is a horizontal asymptote at $1/f_0 \rightarrow 131.665$.}
        \label{fig:classical_vs_Grover_1f}
    \end{figure}
\end{itemize}
For (any variant of) the Louvain algorithm, we know from numerical results that the number of good vertices roughly decreases monotonically during a single iteration of the first phase of the algorithm. We can use this knowledge to set a criterion for $1/f$ beyond which it becomes more efficient to skip the classical sampling step all together and set $N_{\text{samples}} = 0$. Since the number of marked items decreases monotonically only approximately, we allow ourselves some wiggle room. As a consequence, in our simulations, we have chosen to use the following settings\footnote{The plots in the results section are insensitive to fine-tuning $N_{\text{samples}}$ beyond the point that we have done.}
\begin{itemize}
    \item We set $N_{\text{samples}} = 130$ at the start of phase 1. (An extra analyses of \cite{ourotherpaper} suggests setting $N_{\text{samples}}$ slightly lower than $\frac{1}{f_0}$ )
    \item The moment we draw 130  consecutive samples without finding a marked item (implying $1/f \gtrapprox 130$), we set $N_{\text{samples}} = 0$ and use only Grover search from this point on.
\end{itemize}

\paragraph{Optimizing the number of classical samples in \textbf{FindFirst}}
\textbf{FindFirst} searches through sets of varying sizes to detect if they contain a marked element. If the size of the set is sufficiently small it is more efficient to classically search through the set (from start to finish) than to use the quantum \textbf{VertexFind} subroutine. This introduces another hyper-parameter $|L|_{\text{switch}}$ that determines when to switch from using classical search to \textbf{VertexFind}, depending on the list size. In our numerical results we found that classical search was in fact always faster than using \textbf{VertexFind} for the graph sizes we studied. This made it impossible to choose a good setting for this hyper-parameter.

To allow for a comparison to be made between classically traversing the list and using \textbf{VertexFind} within \textbf{FindFirst}, we choose to set $|L|_{\text{switch}}$ to a finite number: $|L|_{\text{switch}} = 512$. I.e., sets with less than $512$ vertices were searched through classically and for sets with more items than $512$ vertices we used \textbf{VertexFind}. The inner loop of \textbf{QLouvain}, searching through the neighbouring communities of a node, was always performed using \textbf{Qsearch} even when we searched through the set classically.

\section{Numerical results}
\label{sec:numerics}
For all numerical results in this section we assume that we require the failure probability of the entire quantum algorithm to be a small constant ($10^{-5}$). We assume that the algorithms make no more than $n \log(n) =: M$ moves (see Section~\ref{sec:num_moves} and Appendix~\ref{sec:number_moves} for justifications for choosing this number of moves), and that the algorithm fails if any subroutine fails within any of these moves. Hence, we set the failure probabilities $\epsilon$ of the main (i.e. outermost) quantum subroutines for each algorithm to
\begin{align}
    \left( 10^{-5}\right)^{\frac{1}{M}} \geq  \frac{10^{-5}}{M}=\frac{10^{-5}}{n \log(n)}=:\epsilon.
\end{align}
Finally, we always use the additional data-structure introduced in Sec.~\ref{sec:num_of_good_vertices} to keep track of the number of good vertices exactly.

\subsection{Artificial data-sets}
\label{sec:query_quantum_classical}
Since we are interested in the run0time scaling as well as the absolute query counts of the classical and quantum Louvain algorithms, we introduce in this section two methods for generating benchmark networks of arbitrary size. One method relies on the well-known LFR method \cite{Lancichinetti2008}, in which node degrees and community sizes are sampled according to power law distributions\footnote{Distributions not uncommon for real-world networks.} with exponents $\tau_1$ and $\tau_2$, respectively. The other input parameters are the total number of nodes $n$, the average degree $\langle d \rangle$ (or alternatively one can set the minimum degree $d_\text{min}$), the maximum degree $d_\text{max}$ (by default set to $n$), minimum and maximum community sizes $S_\text{min}$ and $S_\text{max}$ (by default set to $d_\text{min}$ and $d_\text{max}$, respectively) and a mixing parameter $\mu$ which specifies the fraction of neighbours of a node that do not belong to the node's own community. We will refer to graphs of this type as `LFR-type graphs' and use the \textit{NetworkX} implementation~\cite{Hagberg2008} as a network generator.

In the second graph generation method, which we will call `FCS-type graphs', we fix the community size $S$ and adopt an algorithm similar to that described in Ref.~\cite{Traag2019}. The graph parameters are now $n$, $\langle d \rangle$ and $S$, and the edges are drawn uniformly at random from all possible edges.  A description of our algorithmic implementation (which runs in time $\mO(|E|)$) for FCS-type graph generation is given in Appendix~\ref{sec:random_graphs}. 

\begin{figure}[htb]
    \centering
    \includegraphics[width=\linewidth]{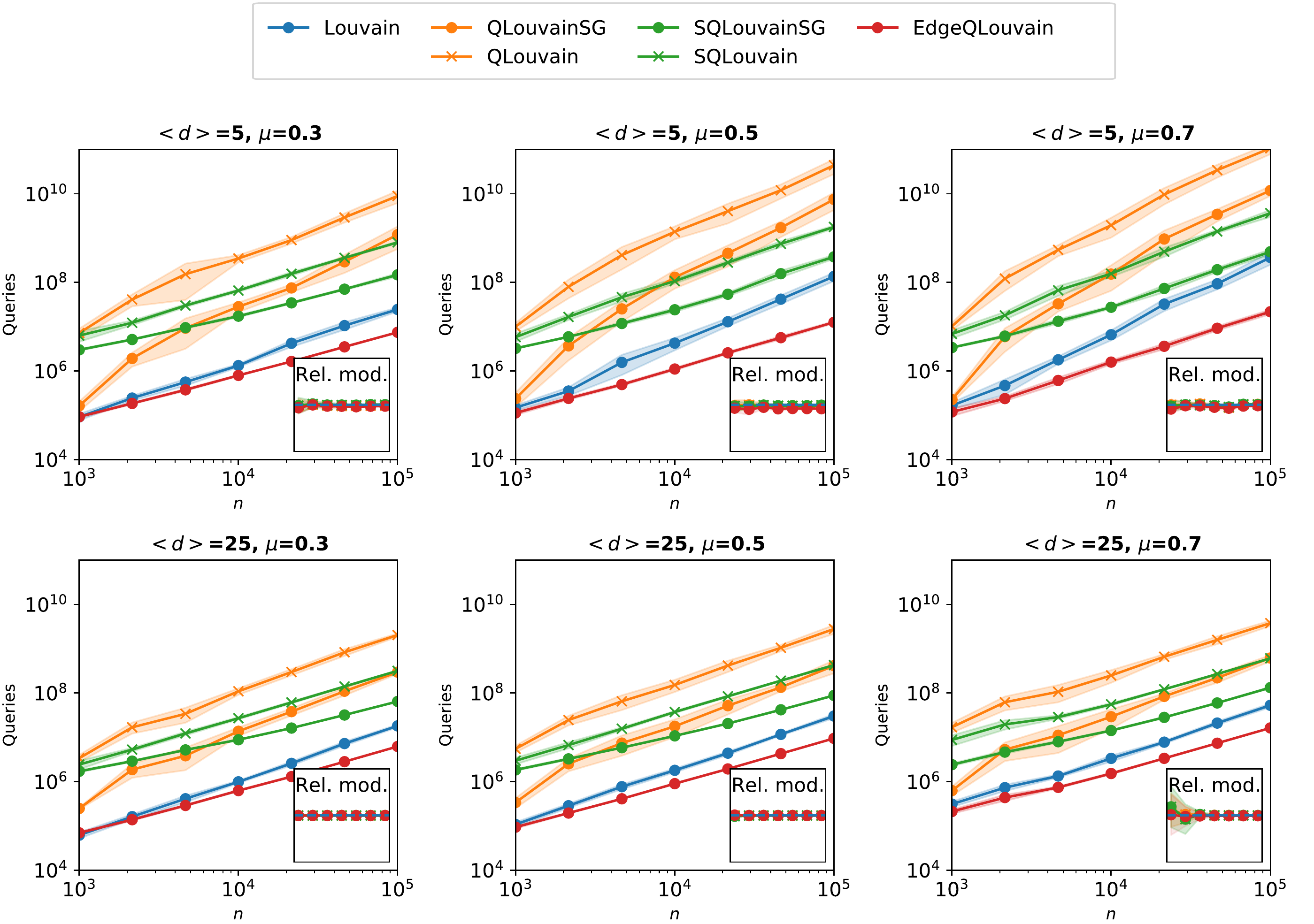}
    \caption{Numerical results for the query counts of the classical and quantum Louvain algorithms from Section~\ref{sec:quantum} on FCS-type graphs with a fixed community size $S=50$. The average degree is either $\langle d \rangle=5$ (top) or $\langle d \rangle=25$ (bottom). The horizontal axis indicates the total number of nodes $n$ and the vertical axis the number of queries made to the function $g_{\Delta}$. Each data point corresponds to the average across 10 randomly generated graphs and the shaded area represents one standard deviation. In every sub-figure the bottom-right box plots the modularity relative to the one obtained with original Louvain (indicated with the dashed blue line) as function of $n$ (logarithmic vertical axis). For the modularity, the limits of the y-axis are set at $\pm 10 \%$ relative difference to Louvain.}
    \label{fig:plots_overview_FCS}
\end{figure}

\begin{figure}[htb]
    \centering
    \includegraphics[width=\linewidth]{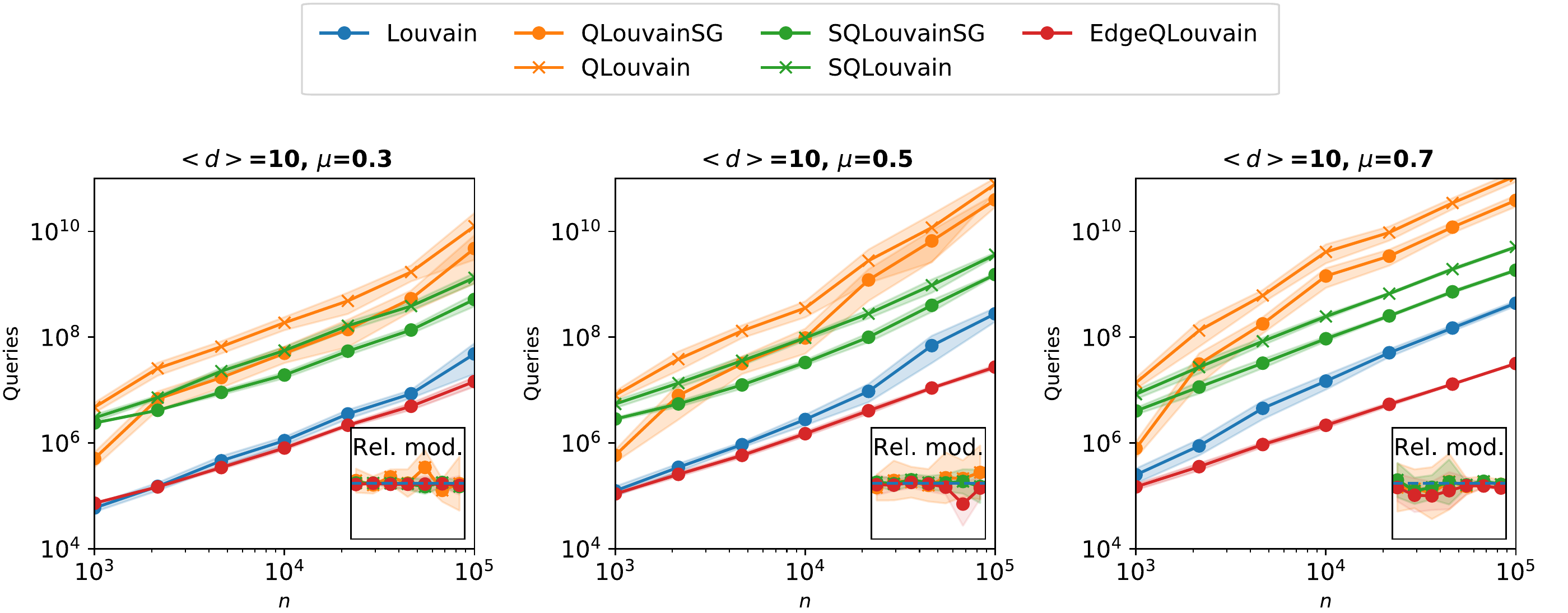}
    \caption{Numerical results for the query counts of the proposed classical and quantum Louvain algorithms on LFR-type graphs with parameters $\tau_1 = 3$, $\tau_2=2$, $\langle d\rangle=10$, $d_\text{max}=100$ and $S_\text{max}=100$. See the caption of Fig. \ref{fig:plots_overview_FCS} for a detailed description of the contents of individual plots and figures.}
    \label{fig:plots_overview_LFR}
\end{figure}

\subsubsection{Absolute query counts}
\label{sssec:absolute_query_counts}
Figs.~\ref{fig:plots_overview_FCS} and~\ref{fig:plots_overview_LFR} show the estimated average number of queries made and the modularities obtained by the classical and all of our quantum Louvain algorithms on FCS- and LFR-type networks with up to $10^5$ nodes. In terms of absolute number of queries, we find that the sparse variants of \textbf{QLouvain} and \textbf{SimpleQLouvain} generally outperform their non-sparse counterparts and that all four are outperformed by the much simpler \textbf{EdgeQLouvain} algorithm. In fact, only \textbf{EdgeQLouvain} was able to achieve an observable speed-up over classical Louvain within the limits of our data sets -- for $n>2000$ it achieves this for all studied graph types and parameter combinations. 

We also find that the quantum algorithms perform better relative to their classical counterpart when $\mu$ is large, which corresponds to graphs with relatively little community structure. An explanation for this is that for those graphs the fraction of good vertices could in general be smaller compared to graphs with high community structure, which corresponds to the regime where \textbf{QSearch} outperforms classical sampling. For the obtained modularities we found that for FCS-type graphs the relative difference in behaviour to the original Louvain algorithm is very small. The differences are more profound for the LFR-type graphs, in particular for \textbf{EdgeQLouvain}.

\subsubsection{Estimating average case polynomial speed-ups}
In Section~\ref{sec:quantum} we showed analytically that, for a large collection of graph configurations, quantum algorithms for community-detection can achieve a polynomial speed-up over the classical Louvain algorithm on which they are based. In this section we use the data of Section~\ref{sssec:absolute_query_counts} to estimate this polynomial speed-up, in expectation, for different graph configurations. Table~\ref{tab:speedup_scaling} shows the coefficients corresponding to the linear fits of log-log plots of the data for the different algorithms as shown in Figs.~\ref{fig:plots_overview_FCS} and~\ref{fig:plots_overview_LFR}. Recall that the ratio of the obtained coefficient as compared to Louvain is upper bounded by $2$ -- achieving this value would correspond to the full quadratic speed-up. For all studied graph configurations, \textbf{SimpleQLouvain} (both sparse and non-sparse) and \textbf{EdgeQLouvain} show (varying) polynomial speed-ups over Louvain. For FCS-type graphs, \textbf{SimpleLouvainSG} achieves the overall best scaling and for LFR-type graphs the best scaling is achieved by \textbf{EdgeQLouvain}. \textbf{QLouvain} almost always has a scaling that is at most comparable to Louvain, but this is mostly an artefact of the algorithm itself: for relatively small instances it predominantly uses classical routines; as $n$ increases it begins to use more Grover steps -- however, these Grover steps are still performed on relatively small lists. We expect that for much larger $n$ one will also observe an asymptotic speed-up for \textbf{QLouvain} and \textbf{QLouvainSG}, though of course these sizes of $n$ might not occur in practice.
\begin{table}[]
\centering
\begin{adjustbox}{width=1.1\textwidth}
\begin{tabular}{cl|l|l|l|l|l|l|l|}
\cline{3-8}
 & \multicolumn{1}{c|}{} & \multicolumn{6}{c|}{\cellcolor[HTML]{C0C0C0}Degree weighted poly-fit $^\text{(Estimated polynomial speed-up)}$} \\  \cline{2-8} 
\multicolumn{1}{c|}{} & \multicolumn{1}{c|}{\cellcolor[HTML]{EFEFEF}Configuration} & \cellcolor[HTML]{EFEFEF}OL & \cellcolor[HTML]{EFEFEF} QLSG& \cellcolor[HTML]{EFEFEF} QL& \cellcolor[HTML]{EFEFEF} SQLSG & \cellcolor[HTML]{EFEFEF}SQL & \cellcolor[HTML]{EFEFEF}EQL  \\ \hline 
\multicolumn{1}{|c|}{\cellcolor[HTML]{EFEFEF}}  & $\langle d \rangle = 5$, $\mu=0.3$ & \ccb{1.23}& \ccrrr{1.76 $^{(0.70)}$}& \ccr{1.45 $^{(0.85)}$}& \ccggg{0.86 $^{(1.43)}$}& \ccg{1.07 $^{(1.15)}$}& \ccgg{0.96 $^{(1.28)}$}\\ \cline{2-8}
\multicolumn{1}{|c|}{\cellcolor[HTML]{EFEFEF}}  & $\langle d \rangle = 5$, $\mu=0.5$ & \ccb{1.5}& \ccrr{2.07 $^{(0.72)}$}& \ccr{1.69 $^{(0.89)}$}& \ccggg{1.07 $^{(1.40)}$}& \ccgg{1.23 $^{(1.22)}$}& \ccggg{1.03 $^{(1.46)}$}\\ \cline{2-8}
\multicolumn{1}{|c|}{\cellcolor[HTML]{EFEFEF}}  & $\langle d \rangle = 5$, $\mu=0.7$ & \ccb{1.71}& \ccrr{2.17 $^{(0.79)}$}& \ccr{1.89 $^{(0.90)}$}& \ccggg{1.12 $^{(1.53)}$}& \ccgg{1.37 $^{(1.25)}$}& \ccggg{1.15 $^{(1.49)}$}\\ \cline{2-8}
\multicolumn{1}{|c|}{\cellcolor[HTML]{EFEFEF}}  & $\langle d \rangle = 25$, $\mu=0.3$ & \ccb{1.24}& \ccr{1.44 $^{(0.86)}$}& \ccr{1.34 $^{(0.93)}$}& \ccggg{0.80 $^{(1.55)}$}& \ccg{1.06 $^{(1.17)}$}& \ccgg{0.98 $^{(1.27)}$}\\ \cline{2-8}
\multicolumn{1}{|c|}{\cellcolor[HTML]{EFEFEF}}  & $\langle d \rangle = 25$, $\mu=0.5$ & \ccb{1.21}& \ccr{1.42 $^{(0.85)}$}& \ccr{1.29 $^{(0.94)}$}& \ccggg{0.85 $^{(1.42)}$}& \ccg{1.08 $^{(1.12)}$}& \ccgg{1.01 $^{(1.20)}$}\\ \cline{2-8}
\multicolumn{1}{|c|}{\multirow{-6}{*}{\cellcolor[HTML]{EFEFEF}\rotatebox{90}{FCS}}}   & $\langle d \rangle = 25$, $\mu=0.7$ & \ccb{1.13}& \ccrr{1.37 $^{(0.82)}$}& \ccb{1.14 $^{(0.99)}$}& \ccgg{0.87 $^{(1.30)}$}& \ccgg{0.92 $^{(1.23)}$}& \ccg{0.95 $^{(1.19)}$}\\ \hline
\multicolumn{1}{|c|}{\cellcolor[HTML]{EFEFEF}}  & $\langle d \rangle = 10$, $\mu=0.3$ & \ccb{1.43}& \ccrr{1.81 $^{(0.79)}$}& \ccr{1.62 $^{(0.88)}$}& \ccg{1.27 $^{(1.13)}$}& \ccb{1.38 $^{(1.04)}$}& \ccg{1.21 $^{(1.18)}$}\\ \cline{2-8}
\multicolumn{1}{|c|}{\cellcolor[HTML]{EFEFEF}}  & $\langle d \rangle = 10$, $\mu=0.5$ & \ccb{1.75}& \ccrr{2.31 $^{(0.76)}$}& \ccr{1.95 $^{(0.90)}$}& \ccg{1.48 $^{(1.18)}$}& \ccg{1.47 $^{(1.19)}$}& \ccgg{1.27 $^{(1.38)}$}\\ \cline{2-8}
\multicolumn{1}{|c|}{\multirow{-3}{*}{\cellcolor[HTML]{EFEFEF}\rotatebox{90}{LFR}}}& $\langle d \rangle = 10$, $\mu=0.7$ & \ccb{1.63}& \ccrr{2.03 $^{(0.80)}$}& \ccr{1.81 $^{(0.90)}$}& \ccg{1.42 $^{(1.15)}$}& \ccg{1.44 $^{(1.13)}$}& \ccgg{1.19 $^{(1.37)}$}\\  \hline
\end{tabular}
\quad
\begin{tabular}{c}
{\cellcolor[HTML]{C0C0C0} \makecell{Speed-up\\ factor}}\\
\rowcolor[HTML]{F35955} 
{$<0.71$} \\
\rowcolor[HTML]{FF8F8C} 
$0.71$ - $0.83$ \\
\rowcolor[HTML]{FFCCC9} 
$0.83$ - $0.95$ \\
\rowcolor[HTML]{DAE8FC} 
 $0.95$ - $1.05$\\
\rowcolor[HTML]{9AFF99} 
{$1.05$ - $1.20$ } \\
\rowcolor[HTML]{87DC86} 
$1.20$ - $1.40$ \\
\rowcolor[HTML]{4EBC4E} 
$>1.40$
\end{tabular}
\end{adjustbox}
    \caption{Estimated polynomial degrees of the expected number of queries for the quantum algorithms and original Louvain on the FSC and LFR-type graphs. The central number in each cell corresponds the estimated polynomial degree obtained from a weighted fit of the form $a n + b$, with the weights set at $\log{n}$, using the log-log data from Figures~\ref{fig:plots_overview_FCS} and~\ref{fig:plots_overview_LFR}. The number in the top-right corner (in parentheses) estimates the polynomial speed-up, and is defined as the ratio of the estimated polynomial degree of Louvain and the respective quantum algorithm.} 
\label{tab:speedup_scaling}
\end{table}

\subsection{Real-world data-sets}
Since analysis based on artificial networks only has limited value in predicting performance on actual real-world networks, we have also performed runs on large data-sets available at~\cite{snapnets} and~\cite{Rossi2015}. Table~\ref{tab:real_queries} shows the obtained modularities and total query count for a selection of our quantum algorithms. Both \textbf{QLouvain} and \textbf{QLouvainSG} are not considered here, as we found they were always outperformed by  \textbf{SimpleQLouvain} and \textbf{SimpleQLouvainSG} on the artificial data sets. We find that for these instances \textbf{EdgeQLouvain} is able to achieve a modest speedup over \textbf{OL}, as well as obtaining slightly better modularities on average.  Similar to the results obtained for artificial networks in Section~\ref{sec:query_quantum_classical}, \textbf{SimpleQLouvainSG} is not able to achieve a speed-up on graphs with sizes of the orders of magnitude considered. Interestingly, in all but one case \textbf{SimpleQLouvain} outperformed \textbf{SimpleQLouvainSG} even though the real data sets have a very low average degree and hence are relatively sparse. This is due to the fact that the considered networks have in fact sometimes a very large maximum degree $d_\text{max}$, which is the relevant graph parameter that determines the query complexity: both \textbf{SimpleQLouvain} and \textbf{SimpleQLouvainSG} scale with $\delta_{\text{max}}$,  however with the former having a quadratic improvement over the latter. The fact that the average degree is much lower than $d_\text{max}$ in these data sets suggests that Variable time amplitude amplification might give a significant improvement (see Appendix~\ref{app:VTAA}). Also, we found that all simulations of our quantum algorithms were not able to finish running within 5 days for the largest considered data set (IMDB). Therefore, further improvements need to be made to our simulation process to be able to study even larger problem instances.

\begin{table}[H]
\centering
\begin{adjustbox}{width=1\textwidth}
\begin{tabular}{|l|l|l|l|l|l|l|l|l|l|l|l|l|}
\hline
\rowcolor[HTML]{EFEFEF} 
 &  & &  & \multicolumn{3}{l|}{\cellcolor[HTML]{EFEFEF}Modularity} & \multicolumn{4}{l|}{\cellcolor[HTML]{EFEFEF}Total queries ($\times 10^7$)} \\ \hline
\rowcolor[HTML]{EFEFEF} 
 & Nodes & Edges & $d_\text{max}$ & OL & SQL(SG) & EQL & OL & SQL & SQLSG  & EQL\\ \hline
Academia \cite{snapnets} & 200k & 1M & 10693 & 0.6447  & 0.6456  & 0.6353 & 3.31 & 267  & 1180 & 1.08  \\ \hline
DBLP \cite{snapnets} & 317k & 1M & 343 &  0.8206  &  0.8210  & 0.8223  & 3.81 & 234 & 309 & 1.34  \\ \hline
Amazon \cite{snapnets} & 335k & 925k & 549 &  0.9262 & 0.9263  & 0.9264   & 2.09 & 195 & 118   & 1.21  \\ \hline
Youtube \cite{snapnets} & 496k & 2M & 25409 & 0.6825 &  0.678  & $-$ & 3.35 & 743 & 6487 & $-$  \\ \hline
IMDB \cite{Rossi2015} & 896K & 4M  & 1590 &  0.6872 & $-$  &  $-$ & 13.1 & $-$ & $-$ & $-$\\ \hline
\end{tabular}
\end{adjustbox}
\caption{Numerical results for the total amount of queries the classical and selected quantum Louvain algorithms make on real-world data-sets, averaged over five different runs.  Entries with `$-$' timed out as they took longer than 5 days to compute.}
\label{tab:real_queries}
\end{table}
\subsection{Conclusion and discussion}
\label{sec:conclusion}

In this paper we considered the framework of~\cite{ourotherpaper} for estimating the run-times of quantum algorithms that achieve modest polynomial speedups over their classical counterparts. As suggested there, in many cases, a traditional asymptotic analysis of the quantum algorithm is not informative enough to make decisions about whether, or for what input sizes, they might achieve a speedup over the best classical algorithm. In some cases, this is because a representative run-time cannot be obtained via such an analysis -- something that is particularly true for (classical and quantum) heuristic algorithms. In others, it may be because the quantum algorithm is particularly sensitive to the input on which it is run, or just that the run-time obtained via an ordinary complexity analysis is not representative of the algorithm's true run-time. 

\ 

\noindent To evaluate the usefulness of the approach we outlined in~\cite{ourotherpaper}, we applied it here to a particular use-case of practical interest: community detection in large networks. Taking as a starting point a popular classical algorithm (the Louvain algorithm), we designed several quantum algorithms, each promising to give \textit{some} speedup over the original. Using the bounds derived in~\cite{ourotherpaper}, we obtained bounds for the quantum subroutines used by these algorithms, and then estimated the complexities of each algorithm for a number of randomly generated graphs, as well as some large real-world ones. We found that the algorithms whose analytically-derived asymptotic complexities were favourable were \textit{not} the algorithms that obtained the lowest complexities in practice, nor the ones that scaled most favourably as a function of input size. This was perhaps not unsurprising, but does demonstrate that an analysis that goes beyond the usual asymptotic complexity one is necessary if one wishes to know whether a particular quantum algorithm could give a speedup for a task of practical interest, on an input representative of the ones it will receive in practice. 

Our main observation when estimating the run-times of our quantum algorithms was that there is a large overhead associated with success probability amplification of quantum sub-routines that is not present in the classical case, and that this can negate the speedup for even very large problem instances. This overhead is made more pronounced by the behaviour of Grover search when there are no marked items in the list: to verify that this is indeed the case, the quantum algorithm must perform a reasonably large number of repetitions of this Grover search sub-routine. Hence, in algorithms where we expect many of lists being searched over to in fact be empty (which was the case for one of our quantum algorithms, \textbf{QLouvain}), the quantum algorithm is often very slow in practice. Interestingly, this behaviour is not made apparent by the usual asymptotic run-time analysis of the algorithm. 

Finally, we note that this kind of empirical run-time analysis is particularly useful for evaluation of quantum speedups of classical heuristics, such as those we consider in this work. In these cases, even if we know that the quantum algorithm achieves a per-step speedup over the classical algorithm, we will not know how much of this speedup survives when the algorithm is run to convergence, suggesting the need for an empirical approach to run-time estimation.

\

\noindent In the context of evaluating the potential of quantum algorithms in real-world settings, we argue that it would be useful to make the sorts of analyses that we perform in this work more commonplace. One way to do this would be to include the option of simulating certain quantum algorithms, in the sense of this paper, within any of the existing quantum programming languages. Most of the work to do this would lie in proving good bounds on the run-times of various quantum primitives. Already for a couple of simple primitive we found this to be an extensive endeavour, but the upshot is that this work would only need to be performed once. With such tools at their disposal, we imagine that it would be easier for quantum algorithms designers to tailor their algorithms to particular tasks and datasets, and to more rapidly prototype ideas for quantum speedups without first undertaking an in-depth mathematical study.

\paragraph{Funding information}
CC was supported by QuantERA project QuantAlgo 680-91-034, with further funding provided by QuSoft and CWI. MF and JW were supported by the Dutch Ministry of Economic Affairs and Climate Policy (EZK), as part of the Quantum Delta NL programme. IN was supported by the DisQover project:  a  collaboration between QuSoft and ABN AMRO, and recieved funding from ABN AMRO and CWI. 

\paragraph{Acknowledgements} We would like to thank Harry Buhrman, Arjan Cornelissen, and Ian Marshall for helpful discussions, and Joran van Apeldoorn for providing tips on using Grover search to find the first item in a list. We Also thank Ton Poppe and Edo van Uitert from ABN AMRO for suggesting to study the use-case of community detection. The numerics of Section~\ref{sec:numerics} were carried out on the Dutch national e-infrastructure with the support of the SURF Cooperative.

\appendix

\section{The number of moves made by the Louvain algorithm}
\label{app:number_of_moves_main_section}

In this section we investigate the number of moves performed by the variants of the Louvain algorithm discussed in this paper as a function of the number of nodes $n$ of the input graph. To start with, we provide a loose upper bound on the number of moves in Section~\ref{app:louvain_total_moves}. Next, we numerically investigate the number of moves performed by the Louvain algorithm on actual datasets in Section~\ref{sec:number_moves}.

\subsection{A bound on the total number of moves}\label{app:louvain_total_moves}
First, we point out that there is a trivial upper bound on how long the Louvain algorithm takes to finish. Write $T$ for the maximum possible number of moves that can be made before the (first phase of the) algorithm terminates. Since modularity is trivially bounded between two constants:
\[
    |Q| \leq \frac{1}{2W} \sum_{u,v \in V} A_{uv} + \frac{1}{4W^2}\sum_{u \in V} s_u \sum_{v\in V} s_v \leq 2,
\]
and each move must strictly increase the modularity, it suffices to bound the smallest amount by which $Q$ can increase after a single move. 

Recall that the change of modularity when moving a vertex $u$ from community $C_{\ell(u)}$ to community $C_{a}$ is
\[
	\Delta_u^a = \frac{S_u^{a} - S_u^{\ell(u)}}{W} - \frac{s_u  \left(\Sigma_{a} - \Sigma_{\ell(u)} + s_u \right)}{2W^2}.
\]
The values in the numerators of the terms are sums over weights of edges, and $W$ is the sum of all weights in the graph. Recall that we have $n = |V|$ vertices. If the weights on the edges are integers with $O(\log n)$-bit representations, then it is clear that the smallest non-zero value of $\Delta_u^a$ for $u \in V$ and $a \in [n]$ is $\Delta_{\min} = 1/W^2 = \frac{1}{|E|^2 \cdot O(\poly(n))}$. For an unweighted graph, this is just $\frac{1}{|E|^2}$. Hence, the maximum number of moves $T$ that can be made before there are no more moves that can increase modularity is $2 / \Delta_{min} = |E|^2 \cdot O(\poly(n)) = O(\poly(n))$. For an unweighted graph, we have in particular $T = O(|E|^2)$. 

Since $\frac{1}{p_k} \leq n$ (for every $k$), we can upper bound the time complexity of the Louvain algorithm by
\begin{itemize}
    \item $\mO(nd|E|^2)$ for unweighted graphs, and
    \item $\mO(\poly(n))$ for weighted graphs with weights expressed with $\log(n)$ bits.
\end{itemize}
We conclude that the Louvain algorithm is at worst a polynomial-time algorithm, although its run-time in practice will depend on the particular problem instance. In practice, the Louvain algorithm is mostly used for large $n$, small (constant) $d$ sparse graphs\footnote{See references in the Introduction.}.

\subsection{The number of moves for actual datasets}
\label{sec:number_moves}

\begin{figure}[htb]
    \centering
    \includegraphics[width=0.6\linewidth]{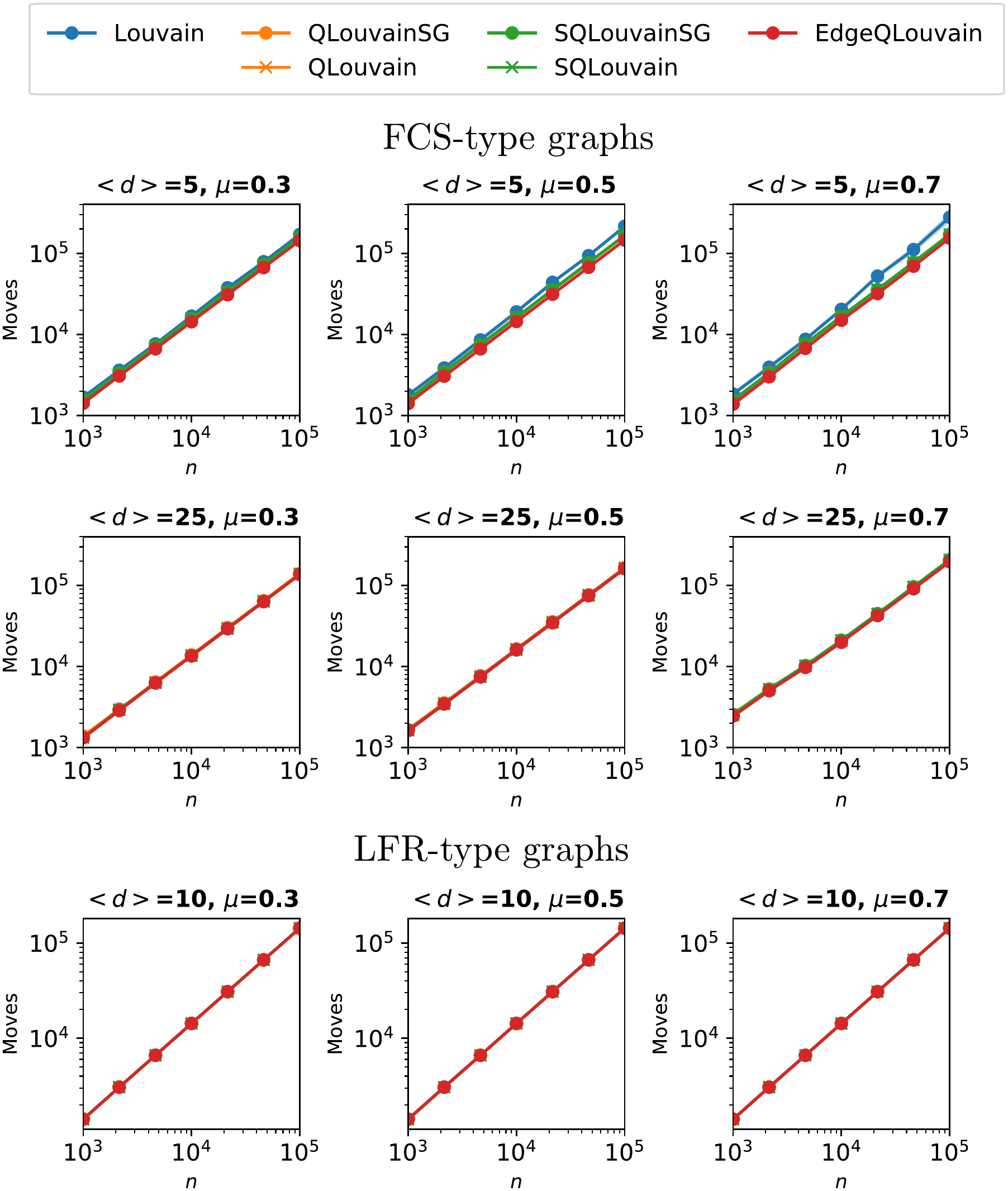}
    \caption{The average number of moves of all considered algorithms as a function of the graph size $n$.}
    \label{fig:moves_overview}
\end{figure}

Next, we numerically investigate the number of moves performed by the variants of the Louvain algorithm discussed in this paper as a function of the number of nodes $n$ of the input graph. Our results can be found in Fig.~\ref{fig:moves_overview}.

Based on Fig.~\ref{fig:moves_overview} we find that for all quantum algorithms the amount of moves scale similarly in graph size $n$ as compared to the Louvain algorithm, with the biggest differences occurring for graphs with a lower average degree $\langle d \rangle$. In agreement with~\cite{lancichinetti2009community}, the data backs up the claim that the number of moves is $\mO(n \log(n))$.

\section{Numerical results on the data structure}
\label{sec:runtime_OL_LL}

In Section~\ref{sec:complexity_of_louvain} we mentioned that adding an extra data structure allows for a run time speedup of OL (original Louvain). In this section we show numerical results supporting this claim.  

We compare the run time between OL and OL with extra data structure using a python implementation of Louvain given by \cite{python_louvain}. For OL we use exactly this implementation of louvain. For OL with extra data structure we changed the code to incorporate the data structure as described in  Section~\ref{sec:complexity_of_louvain}. Run time tests are done in real time by tracking how long the algorithm takes to converge. The algorithms should in principle converge to a similar solutions since they differ only in their internal randomness. We also do a memory test, to see how much extra memory is used by introducing the extra data structure. This memory test registers the peak usage of memory. 

\begin{figure}[htbp]
    \centering
    \includegraphics[width=\linewidth ]{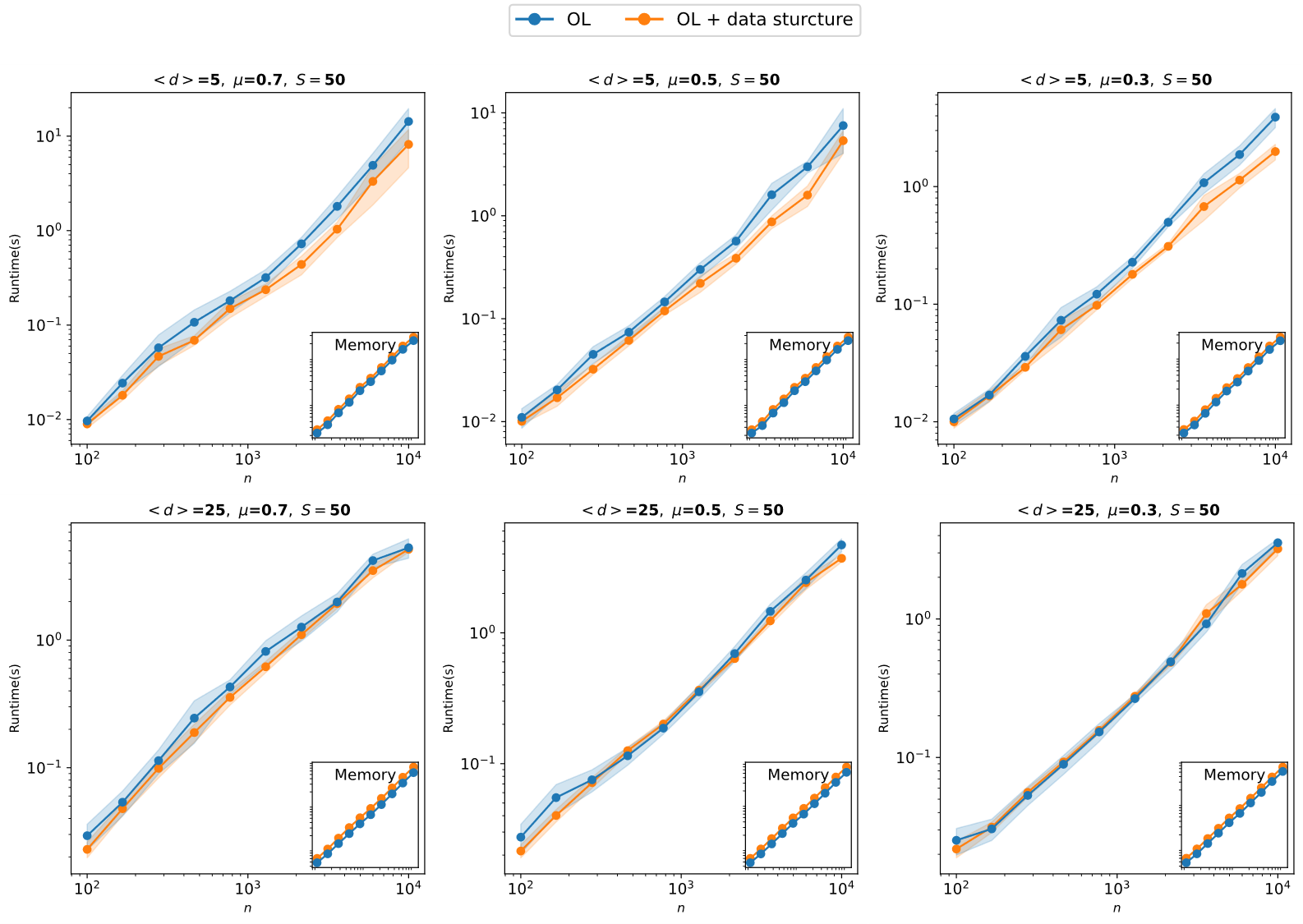}
    \caption{Numerical results for run time comparison between original Louvain and original Louvain with extra data structure. The algorithms were tested on FCS-type graphs with fixed community size $S=50$, and average degrees $\langle d\rangle = 5$ (top) and $\langle d\rangle = 25$ (bottom). From left to right community structure is increased by decreasing $\mu$. The horizontal axis indicates the amount of nodes $n$ and the vertical axis the total run time (top). Each point corresponds to an average over 10 randomly generated graphs and the shaded area represents the standard deviation. The box in the bottom right shows peak memory usage (the number of bits).}
    \label{fig:runtime}
\end{figure}

The algorithms are tested on FCS- and LFR-type graphs as generated by the algorithm in Appendix~\ref{sec:random_graphs} and the NetworkX implementation~\cite{Hagberg2008}. Every instance is run on $10$ graphs, averaged and shown in Figure \ref{fig:runtime}. As predicted, OL with data structure slightly outperforms OL in run time when the average degree is low ($\langle d\rangle = 5$). This run time advantage is lost when the average degree increases ($\langle d\rangle = 25$). There is an extra constant overhead in memory when the extra data structure is introduced, as expected.

\section{Quantum Louvain with variable time amplitude amplification}\label{app:VTAA}
In this appendix we describe in detail our quantum algorithm for community detection based on the technique of variable time amplitude amplification, and give a proof of its correctness and run-time. In~\cite{ambainis12}, Ambainis describes a version of amplitude amplification for the situation where the subroutine used by the quantum algorithm takes a different time to finish for each branch. This algorithm goes by the name variable time amplitude amplification (VTAA). Intuitively, VTAA executes amplitude amplification with a subroutine that can have different stopping times, with a final run-time that depends on some average of the individual stopping times, rather than being limited by the slowest branch as is the case in the algorithm presented in the previous section. We note that it is also possible to use variable time \emph{Grover search}\footnote{Here, we have a collection of $n$ items $x_1,\dots,x_n$ and we would like to find an $i : x_i = 1$. Let $t_i$ be the number of time steps required to evaluate each $x_i$. Then variable time Grover search can find an $i$ in time $\tilde{O}(\sqrt{t_1 + t_2 + \cdots + t_n})$. If there are multiple $i$'s satisfying $x_i=1$, then the algorithm actually becomes \emph{slower} (by a constant factor).}~\cite{ambainis2010quantum} to obtain this behaviour, but here the algorithm assumes that there is only a single marked item and the run-time does not improve when there are multiple marked items, and hence the run-time of the algorithm will generally be quite poor, taking a time that is roughly $\tilde{O}(\sqrt{\delta_{\avg} |L|})$ per move compared to the classical $O(\delta_{\avg} / f)$, with $k$ the fraction of good vertices present in $L$.

\ 

\noindent Here we use VTAA with a classical subroutine $\mA$ that checks the neighbouring communities of a vertex one by one, requiring $\delta_u$ calls to the unitary for $g_\Delta$ and other operations for each vertex $u$. Let $L$ be a list of vertices, and let $n \geq |L|$ be an upper bound on its size. The algorithm $\mA$ acts on four registers, a $\log(n)$-sized vertex register, two $\log(\delta_{\max})$-sized neighbor-index registers, and 2-qubit flag register that can take the values 0, 1 and 2. The flag states correspond to: 0 \emph{found no neighbor to move to}, 1 \emph{found a neighbor to move to}, and 2 \emph{still searching}. $\mA$ consists of the sequence of unitaries: $\mA = \underbrace{\mA_c \cdots  \mA_c}_{\delta_{\max} \,\, \text{times}} \mA_s$, where
\begin{equation}\label{eq:A_s}
    A_s \ket{0} \ket{0} \ket{0} \ket{0} = \frac{1}{\sqrt{|L|}} \sum_{u\in L} \ket{u} \ket{1} \ket{0} \ket{2}
\end{equation}
sets up the initial state, and then each of the remaining $\mA_c$'s sequentially check the neighboring communities of all vertices in superposition. The third register, which is initially set to 0, keeps track of the neighboring community index that currently maximises $\Delta_u^{\eta_u(j)}$. Here we use the convention that $\Delta_u^{\eta_u(0)} = 0$. Now, $\mA_c$ acts on basis states as
\begin{equation}\label{eq:A_c}
    \mA_c \ket{u} \ket{j} \ket{j_{\max}} \ket{f} := 
    \begin{cases}
        \ket{u} \ket{j} \ket{j} \ket{2} & \text{if} \quad f = 2, \quad j < \delta_u \quad \text{and} \quad \Delta_u^{\eta_u(j)} > \Delta_u^{\eta_u(j_{\max})}  \\
        \ket{u} \ket{j+1} \ket{j_{\max}} \ket{2} & \text{if} \quad f = 2, \quad j < \delta_u \quad \text{and} \quad \Delta_u^{\eta_u(j)} \leq \Delta_u^{\eta_u(j_{\max})}  \\
        \ket{u} \ket{j} \ket{j} \ket{1} & \text{if} \quad f = 2, \quad j = \delta_u \quad \text{and} \quad \Delta_u^{\eta_u(j)} > \Delta_u^{\eta_u(j_{\max})}  \\
        \ket{u} \ket{j} \ket{j_{\max}} \ket{1} & \text{if} \quad f = 2, \quad j = \delta_u, \quad \Delta_u^{\eta_u(j)} \leq \Delta_u^{\eta_u(j_{\max})}  \\
        & \text{and} \quad  \Delta_u^{\eta_u(j_{\max})} > 0 \\
        \ket{u} \ket{j} \ket{j_{\max}} \ket{0} & \text{if} \quad f = 2, \quad j = \delta_u, \quad \Delta_u^{\eta_u(j)} \leq \Delta_u^{\eta_u(j_{\max})}  \\
        & \text{and} \quad  \Delta_u^{\eta_u(j_{\max})} = 0 \\
        \ket{u} \ket{j} \ket{j_{\max}}  \ket{f} & \text{if} \quad f = 0 \quad \text{or} \quad f = 1,
    \end{cases}
\end{equation}
where $\eta_u(j) \in \zeta_u$ is the $j$-th neighboring community of $u$. 

The algorithm $\mA_c$ is just a coherent implementation of a classical algorithm that, given a vertex $u$, computes $\Delta_u^{\alpha}$ for the neighboring communities $\alpha \in \zeta_u$ of $u$ one by one. After visiting all neighboring communities, $\mA_c$ stops and puts in the fourth register either a $0$, signifying that no neighbouring communities of $u$ offer a good move, or otherwise puts a $1$. Because different vertices will have a different number of neighboring communities, $\mA_c$ has several different stopping times. 

If we were to run $\mA$ on the all zeros state, we would obtain the final state
\begin{eqnarray}
    \mA \ket{0}\ket{0}\ket{0}\ket{0} &=& \frac{1}{\sqrt{|L|}} \sum_{u \in L} \ket{u}\ket{\delta_u}\ket{j_u}\ket{\bar{\Delta}_u > 0?} \\
    &=& \alpha_0 \ket{\psi_0}\ket{0} + \alpha_1 \ket{\psi_1}\ket{1} =: \ket{\psi_{\text{final}}}\,,
\end{eqnarray}
where $\ket{\psi_b}$ is the state on the first three registers corresponding to the branch of the superposition in which the flag register is in state $\ket{b}$. The goal of amplitude amplification is to amplify the part of the state in which the flag register is in state $\ket{1}$, i.e. to amplify the amplitude $|\alpha_1|$ to $\geq 1/\sqrt{2}$, since this part of the superposition contains the indices of vertices for which a good move is available. The goal of \emph{variable time} amplitude amplification is to achieve this amplification using a number of applications of the constituent unitaries $\cA_c$ that takes into account the different stopping times for $\cA$ on different branches.

Hence, our approach is to use the algorithm of Lemma~\ref{lem:VTAA}, in which case the set of different stopping times $\{t_1,\ldots, t_m\}$ appearing in Eq.~\eqref{eq:stopping_times} for $\mA_c$ will be the set of numbers of neighboring communities $\{\delta_u: u\in L\}$, and $T_{\max} = \delta_{\max}$ (this requires us to know $\delta_{\max}$ before applying the algorithm, but we note that this can easily be kept track of and updated after every vertex move). In order to be able to use VTAA, we need to check that $\mA_c$ meets the necessary requirements outlined in~\cite{ambainis2010quantum}. To this end, for $i \in [\delta_{\max}]$ define $\mH_i := \text{Span} \left( \{ \ket{u} \ket{j} \ket{*}: u\in L, 1\leq j \leq i \} \right)$, where $\ket{*}$ means there is no condition on the $j_{\max}$ register. Note that the subspaces $\mH_i$ do not involve the flag register. Then we can prove Theorem~\ref{theo:VTAA}, which is restated below for convenience.

\VTAA*

\begin{proof}
We will use the unitaries $\mA_c$ and $\mA_s$ defined above to construct the algorithm. In order to be able to apply Lemma~\ref{lem:VTAA}, we first must check that $\mA_c$ and $\mA_s$ satisfy the various conditions described in~\cite{ambainis12}. In particular, we need to check that the following hold:
\begin{enumerate}
    \item For $i \in [\delta_{\max}-1]$, $\mH_i \subseteq \mH_{i+1}$. 
    \item For $i \in [\delta_{\max}]$, we should have that the state $\ket{\psi_i}$ obtained after $i$ applications of $\mA_c$ can be expressed as 
    \[
        \ket{\psi_i} = \underbrace{\mA_c \cdots  \mA_c}_{i \,\, \text{times}} \mA_s \ket{0} \ket{0} \ket{0} =
        \alpha_{i,0} \ket{\psi_{i,0}} \ket{0} + \alpha_{i,1} \ket{\psi_{i,1}} \ket{1} + \alpha_{i,2} \ket{\psi_{i,2}} \ket{2},
    \]
    where $\ket{\psi_{i,0}} \in \mH_i$, $\ket{\psi_{i,1}} \in \mH_i$, and $\ket{\psi_{i,2}} \in (\mH_i)^\perp$.
    \item For $i\in [\delta_{\max}]$ and $P_{\mH_i}$ the projector onto space $\mH_i$, we have
    \begin{equation}
        P_{\mH_i} \ket{\psi_{i+1,0}} = \ket{\psi_{i,0}} \quad \text{and}
        \quad P_{\mH_i} \ket{\psi_{i+1,1}} = \ket{\psi_{i,1}}.
        \label{eq:AcProj}
    \end{equation}
\end{enumerate}

These conditions clearly hold for $\mA_s$, and so we will focus on the unitary $\mA_c$. Condition 1 holds by definition of the subspaces $\{\mH_i\}_{i\in [\delta_{\max}]}$. In order to verify that condition 2 holds, given $i \in [\delta_{\max}]$, we observe the following.
\begin{itemize}
    \item $\ket{\psi_{i,0}}$ is a superposition over vertices $u$ for which $\delta_u \leq i$ and $\Delta_u^{\eta_u(j_{\max})} = 0$. In particular, for every $u$ in the superposition, its neighboring community index is set to $\delta_u \leq i$, and  hence $\ket{\psi_{i,0}} \in \mH_i$.  
    \item $\ket{\psi_{i,1}}$ is a superposition over vertices for which $\delta_u \leq i$ such that $\bar{\Delta}_u = \Delta_u^{\eta_u(j_{\max})} > 0$. In particular, for every vertex $u$ in the superposition, its neighbor index is set to $\delta_u \leq i$ and hence $\ket{\psi_{i,1}} \in \mH_i$.  
    \item  $\ket{\psi_{i,2}}$ is a superposition over vertices $u$ for which $\delta_u > i$ (otherwise the flag would have been set to 0 or 1). In particular, all vertices $u$ in the superposition have their neighbor index set to $i+1$, and therefore$\ket{\psi_{i,2}} \in (\mH_i)^\perp$. 
\end{itemize}
For condition 3, we notice that when we apply $\mA_c$ to $\ket{\psi_i}$, $\mA_c$ only acts on $\ket{\psi_{i,2}}$, and then sets some flags for vertices $u$ in the superposition $\ket{\psi_{i,2}}$ to 0 or 1. The vertices for which the flag was set to 0 or 1 all have their neighbor index set to $i+1$. Thus, when we apply the projection opertor $P_{\mH_i}$ to $\ket{\psi_{i+1,0}}$ or $\ket{\psi_{i+1,1}}$, these newly added vertices project to 0, and therefore Eq.~\eqref{eq:AcProj} is satisfied.

\ 

Having verified that our subroutines $\cA_c$ and $\cA_s$ can be used inside VTAA, we turn our attention to the complexity of the resulting algorithm. Recall that $f$ is the fraction of vertices in $L$ that have a good move available (i.e. $\bar{\Delta} > 0$), and also let $p_i$ be the probability that, for a randomly chosen vertex $u$, the number of neighbouring communities of $u$ is $i$. Then the $l_2$ average over stopping times of $\cA_c$ is 
\[
    t^q_{\avg} = \sqrt{\sum_{i=1}^{\delta_{\max}} p_i i^2}\,
\]
and hence, using Lemma~\ref{lem:VTAA}, we can apply VTAA directly to $\cA_c$ and $\cA_s$ to obtain an algorithm that produces the state
\[
    \ket{\psi_{\delta_{\max}}} = \alpha_{\delta_{\max},0} \ket{\psi_{\delta_{\max},0}} \ket{0} + \alpha_{\delta_{\max},1} \ket{\psi_{\delta_{\max},1}} \ket{1}
\]
with $|\alpha_{\delta_{\max}, 1}|^2 \geq \frac{1}{2}$, by invoking $\cA_c$ and its inverse at most 
\[
    O\left(\delta_{\max} \log(\delta_{\max}) + \frac{t_{\text{avg}}^\text{q}}{\sqrt{f_k}} \log^{1.5}\delta_{\max}\right)
\]
times. By measuring the second register, we will project the first register onto $\ket{\psi_{\delta_{\max},0}}$ with probability $\geq 1/2$, at which point we can measure it to obtain the identity of a good vertex, and the best move available to it, selected at random from the set of all good vertices. By repeating this process $O(\log (1/\epsilon))$ times, we will obtain a good vertex with probability $\geq 1-\epsilon$.

Finally, we note that every application of $\cA_c$ requires $O(1)$ function calls to $g_\Delta$ (whilst $\cA_s$ doesn't require any), and hence the number of function calls made by the algorithm is the same (up to constant multiplicative overhead) as the number of uses of $\cA_c$ and its inverse. 
\end{proof}

Using \textbf{VertexFindVTAA}, we can proceed to construct new VTAA-based versions of the algorithms for community detection described above. As an example, we can construct an analogue to \textbf{SimpleQLouvain}, whose run-time will now become
\[
    \sum_{k \in [T]} \tilde{O}\left(\delta_{\max} + \frac{t_{\text{avg}}^\text{q}}{\sqrt{f_k}} \right)\,,
\]
where the $t_{\text{avg}}^\text{q}$ is as in Theorem~\ref{theo:VTAA}, and we choose $\zeta \leq 1/(3T)$ as the failure probability of \textbf{VertexFindVTAA}, to ensure that every step of the algorithm succeeds with high probability. 

\ 

In contrast, a classical algorithm that searches for good vertices with replacement will make 
\[
    t_{\text{avg}}^\text{c} = \sum_{i=1}^{\delta_{\max}} p_i i
    \label{eq:delta_avg}
\]
function calls per move on average, leading to a classical run-time of 
\begin{equation}
    \sum_{k \in [T]} O\left(\frac{\delta_{\avg}}{f_k} \right).
\end{equation}

The $t^q_{\avg}$ appearing in the quantum complexity is the `2-norm average' of the stopping times, rather than the 1-norm average of Eq.~\eqref{eq:delta_avg} that appears in the classical complexity. If the number of neighbouring communities is constant, then $t^{\text{q}}_{\avg} = \delta_{\avg}$.

\section{Generation of FCS-type graphs}

Algorithm~\ref{alg:FCS} describes the algorithm we use to generate FCS-type random graphs.

\label{sec:random_graphs}
\begin{algorithm}[H]
  \caption{FCS-type graph generation}
  \label{alg:graph_generation1}
   \begin{algorithmic}[1]
    \Function{GraphGenerationFCS}{$n$, $S$, $\mu$, $\langle d\rangle$}
    \State Initialize a graph $G=(V,E)$ where $V=\{1,\dots,n\}$ and $E=\emptyset$. Define community labels $L=\{1,\dots,\lceil n/S  \rceil \}$. 
    \State Set $l_u=u\mod{S}$ for all $u\in \{1,\dots,\lfloor n/S \rfloor\}$, and set $l_u = \lceil n/S \rceil $ for all $u\notin \{1,\dots,\lfloor n/S \rfloor\}$.  Let $V_l\subset V$ be the set of nodes in community $l$. 
    \State Set $k = \langle d\rangle n$ as the counter of the remaining edges to be added. 
    \While{$k>0$}
        \State  pick $l\in L$ randomly uniform
        \State  pick $u\in V_l$ randomly uniform, with probability $1 - \mu$ pick $v$ from $V_l$ uniformly at random , and with probability $\mu$ pick $v$ from  $V\setminus V_l$ uniformly at random.
        \If{$(u,v) \notin E$}
            \State $E \gets E \cup \{(u,v)\} $
            \State  $k \gets k-1$
        \EndIf
    \EndWhile
    \State \Return $\mathcal{G}=(V,E)$
    \EndFunction
    \end{algorithmic}
    \label{alg:FCS}
\end{algorithm}

\bibliography{references.bib}
\bibliographystyle{plain}

\end{document}